\documentclass[a4paper,USenglish,cleveref,autoref,thm-restate, numberwithinsect]{lipics-v2021}

\newbool{fullversion}
\booltrue{fullversion}

\usepackage[appendix=inline,bibliography=common]{apxproof}\usepackage[full]{optional}

\usepackage[textsize=scriptsize]{todonotes}

\setuptodonotes{fancyline}

\usepackage[utf8]{inputenc}
\usepackage{amsthm}
\usepackage{amsfonts}
\usepackage{amsmath}
\usepackage{graphicx}
\usepackage{svg}
\usepackage{amssymb}
\usepackage{tikz}

\usepackage{subcaption}
\usepackage[numbers]{natbib}

\usepackage{etoolbox} 

\usepackage{algorithmic}
\usepackage{algorithm}

\newcommand{\reach}{\Rightarrow}
\renewcommand{\vec}[1]{\ensuremath{\mathbf{#1}}}

\newcommand{\reachset}[1]{\ensuremath{\mathsf{reach}(#1)}}

\newcommand{\C}{\ensuremath{\mathcal{C}}}
\newcommand{\D}{\ensuremath{\mathcal{D}}}
\newcommand{\E}{\ensuremath{\mathcal{E}}}

\newcommand{\N}{\ensuremath{\mathbb{N}}}
\newcommand{\Z}{\ensuremath{\mathbb{Z}}}
\newcommand{\R}{\ensuremath{\mathbb{R}}}
\newcommand{\Q}{\ensuremath{\mathbb{Q}}}
\newcommand{\Rp}{\ensuremath{\mathbb{R}_{\geq 0}}}

\usepackage[T1]{fontenc}

\newcommand{\calD}{\mathcal{D}}

\newcommand{\calC}{\mathcal{C}}
\newcommand{\calE}{\mathcal{E}}

\newcommand{\va}{\vec{a}}
\newcommand{\vc}{\vec{c}}
\newcommand{\vd}{\vec{d}}

\newcommand{\vr}{\vec{r}}
\newcommand{\vv}{\vec{v}}

\newcommand{\vs}{\vec{s}}

\newcommand{\vM}{\vec{M}}

\newcommand{\vp}{\vec{p}}
\newcommand{\vi}{\vec{i}}
\newcommand{\vx}{\vec{x}}
\newcommand{\vy}{\vec{y}}
\newcommand{\vw}{\vec{w}}

\newcommand{\vh}{\vec{h}}

\newcommand{\vu}{\vec{u}}

\newcommand{\vb}{\vec{b}}

\def\longrightharpoonup{\relbar\joinrel\rightharpoonup}
\def\longleftharpoondown{\leftharpoondown\joinrel\relbar}
\def\longrightleftharpoons{\mathop{\vcenter{\hbox{\ooalign{\raise1pt\hbox{$\longrightharpoonup\joinrel$}\crcr\lower1pt\hbox{$\longleftharpoondown\joinrel$}}}}}}
\def\rxn{\mathop{\rightarrow}\limits}  

\def\revrxn{\mathop{\rightleftharpoons}\limits}

\newtheoremrep{thm}[theorem]{Theorem}
\newtheoremrep{lemma}[theorem]{Lemma}
\newtheoremrep{observation}[theorem]{Observation}
\newtheoremrep{corollary}[theorem]{Corollary}

\newcommand{\recc}{\mathrm{recc}}

\nolinenumbers
\pdfoutput=1
\hideLIPIcs
\title{The computational power of discrete chemical reaction networks with bounded executions}
\titlerunning{Execution bounded chemical reaction networks}
\author{David {Doty}}{Computer Science, University of California--Davis, CA, USA \and \url{https://web.cs.ucdavis.edu/~doty/}}{doty@ucdavis.edu}{https://orcid.org/0000-0002-3922-172X}{NSF awards 2211793, 1900931, 1844976, and DoE EXPRESS award SC0024467}
\author{Ben {Heckmann}}{CIT, Technical University of Munich, Germany \and Computer Science, University of California--Davis, CA, USA}{ben.heckmann@tum.de}{}{NSF award 1844976}
\authorrunning{D. {Doty} and B. {Heckmann}}
\Copyright{David {Doty} and Ben {Heckmann}}
\ccsdesc{Theory of computation~Models of computation}
\keywords{chemical reaction networks, population protocols, stable computation}
\date{}

\EventEditors{Dan Alistarh}
\EventNoEds{1}
\EventLongTitle{38th International Symposium on Distributed Computing (DISC 2024)}
\EventShortTitle{DISC 2024}
\EventAcronym{DISC}
\EventYear{2024}
\EventDate{October 28--November 1, 2024}
\EventLocation{Madrid, Spain}
\EventLogo{}
\SeriesVolume{319}
\ArticleNo{19}
\relatedversion{}\relatedversiondetails{Full Version}{https://arxiv.org/abs/2405.08649}

\begin{document}

\maketitle

\begin{abstract}
    Chemical reaction networks (CRNs) model systems where molecules interact according to a finite set of \emph{reactions} such as $A + B \to C$, representing that if a molecule of $A$ and $B$ collide, they disappear and a molecule of $C$ is produced. 
    CRNs can compute Boolean-valued predicates $\phi:\N^d \to \{0,1\}$ and integer-valued functions $f:\N^d \to \N$; for instance $X_1 + X_2 \to Y$ computes the function $\min(x_1,x_2),$
    since starting with $x_i$ copies of $X_i$, eventually $\min(x_1,x_2)$ copies of $Y$ are produced.
    
    We study the computational power of \emph{execution bounded} CRNs,
    in which only a finite number of reactions can occur from the initial configuration
    (e.g., ruling out reversible reactions such as $A \revrxn B$).
    The power and composability of such CRNs depends crucially on some other modeling choices that do not affect the computational power of CRNs with unbounded executions, namely whether an initial leader is present, and whether (for predicates) all species are required to ``vote'' for the Boolean output.
    If the CRN starts with an initial leader, and can allow only the leader to vote,
    then all semilinear predicates and functions can be stably computed in $O(n \log n)$ parallel time by execution bounded CRNs.
    
    However, 
    if no initial leader is allowed, all species vote,
    and the CRN is ``non-collapsing''
    (does not shrink from initially large to final $O(1)$ size configurations),
    then execution bounded CRNs are severely limited, 
    able to compute only \emph{eventually constant} predicates.
    A key tool is a characterization of execution bounded CRNs as precisely those with a nonnegative \emph{linear potential function} that is strictly decreased by every reaction~\cite{czerner2024fast}.
\end{abstract}

\section{Introduction}
Chemical reaction networks (CRNs) are a fundamental tool for understanding and designing molecular systems. By abstracting chemical reactions into a set of finite, rule-based transformations, CRNs allow us to model the behavior of complex chemical systems. 
For instance, the CRN with a single reaction $2X \rightarrow Y$, produces one $Y$ every time two $X$ molecules randomly react together, effectively calculating the function $f(x) = \lfloor x/2 \rfloor$ if the initial count of $X$ is interpreted as the input and the eventual count of $Y$ as the output.
A commonly studied special case of CRNs is the \emph{population protocol} model of distributed computing \cite{angluin2007computational},
in which each reaction has exactly two reactants and two products, e.g., $A+B \to C+D$.
This model assumes idealized conditions where reactions can proceed indefinitely, constrained only by the availability of reactants in the well-mixed solution.

Precisely the \emph{semilinear} predicates $\phi: \N^d \to \{0,1\}$~\cite{angluin2004computation} and functions $f:\N^d \to \N$~\cite{Chen2012DeterministicFunction} can be computed \emph{stably}, roughly meaning that the output is correct no matter the order in which reactions happen.
In population protocols or other CRNs with a finite reachable configuration space, this means that the output is correct with probability 1 under a stochastic scheduler that picks the next molecules to react at random.
However, existing constructions to compute semilinear predicates and functions use CRNs with \emph{unbounded executions},
meaning that it is possible to execute infinitely many reactions from the initial configuration.
CRNs with \emph{bounded executions} have several advantages.
With an absolute guarantee on how many reactions will happen before the CRN terminates,
wet-lab implementations need only supply a bounded amount of fuel to power the reactions.
Such CRNs are simpler to reason about:
each reaction brings it ``closer'' to the answer.
They also lead to a simpler definition of stable computation than is typically employed: an execution bounded CRN stably computes a predicate/function if it gets the correct answer after sufficiently many reactions.


To study this topic, we study networks that must eventually reach a configuration where no further reactions can occur, regardless of the sequence of reactions executed. 
This restriction is nontrivial because the techniques of \cite{Chen2012DeterministicFunction,Doty2013Leaderless} rely on reversible reactions 
(leading to unbounded executions) 
catalyzed by species we expect to be depleted once a computational step has terminated. 
This trick seems to add computational power to our system by undoing certain reactions as long as a specific species is present. Consider the following CRN computing $f(x_1, x_2, x_3) = \min \left(x_1-x_2, x_3\right)$. The input values $x_i$ are given by the counts of $X_i$, and the output by the count of $Z$ molecules in the stable state:
\begin{align} 
    X_1 & \rightarrow Y \label{x_to_y}\\ 
    X_2+Y & \rightarrow \varnothing \label{y_minus_x2}\\ 
    Y+X_3 & \rightarrow Z \label{min_y_x3}\\ 
    Z+X_2 & \rightarrow X_2+X_3+Y \label{reverse_min_y_x3}
\end{align}
Reactions \eqref{x_to_y} and \eqref{y_minus_x2} compute $x_1 - x_2$, storing the result in the count of $Y$. Next, reaction \eqref{min_y_x3} can be applied exactly $\min(y, x_3)$ times. But since the order of reactions is a stochastic process, we might consume copies of $Y$ in \eqref{min_y_x3}, before all of $x_2$ is subtracted from it. Therefore, we add reaction \eqref{reverse_min_y_x3}, using $X_2$ as a catalyst to undo reaction \eqref{min_y_x3} as long as copies of $X_2$ are present, indicating that the first step of computation has not terminated.
However, this means the above CRN does not have bounded exectutions, since reactions \eqref{min_y_x3} and \eqref{reverse_min_y_x3} can be alternated in an infinite execution.
A similar technique is used in \cite{Chen2012DeterministicFunction}, where semilinear sets are understood as a finite union of linear sets, shown to be computable in parallel by CRNs.
A reversible, catalyzed reaction finally converts the output of one of the CRNs to the global output.
Among other questions, we explore how the constructions of \cite{Chen2012DeterministicFunction} and \cite{Doty2013Leaderless} can be modified to provide equal computational power while guaranteeing bounded execution. 

The paper is organized as follows.
\cref{sec:execution-bounded-crns} defines execution boundedness (\Cref{def:execution-bounded}). 
We introduce alternative characterizations of the class for use in later proofs, such as the lack of self-covering execution paths.
\cref{sec:positive-results-sets} and \ref{sec:positive-results-fuctions} contain the main positive results of the paper and provide the concrete constructions used to decide semilinear sets and functions using execution bounded CRNs whose initial configurations contain a single leader.
\cref{sec:limitations} discusses the limitations of execution bounded CRNs, introducing the concept of a ``linear potential function'' as a core characterization of these systems.
We demonstrate that entirely execution bounded CRNs that are leaderless and non-collapsing (such as all population protocols),
can only stably decide trivial semilinear predicates: the \emph{eventually constant} predicates
(\Cref{defn:eventually_constant_predicate}).

\section{Preliminaries}

We use established notation from \cite{Chen2012DeterministicFunction, Doty2013Leaderless} and stable computation definitions from \cite{angluin2007computational} for (discrete) chemical reaction networks.

\subsection{Notation}

Let $\N$ denote the nonnegative integers. For any finite set $\Lambda$, we write $\N^\Lambda$ to mean the set of functions $f: \Lambda \rightarrow \N$. 
Equivalently, $\N^\Lambda$ can be interpreted as the set of vectors indexed by the elements of $\Lambda$, and so $\vc \in \N^\Lambda$ specifies nonnegative integer counts for all elements of $\Lambda$.
$\vc(i)$ denotes the $i$-th coordinate of $\vc$, and if $\vc$ is indexed by elements of $\Lambda$, then $\vc(Y)$ denotes the count of species $Y \in \Lambda$.
We sometimes use multiset notation for such vectors, e.g., $\{A, 3C\}$ for the vector $(1,0,3)$, assuming there are three species $A,B,C$.
If $\Sigma \subseteq \Lambda$, then $\vec{i} \upharpoonright \Sigma$ denotes restriction of $\vec{i}$ to $\Sigma$.

For two vectors $\vx,\vy \in \R^k$,
we write 
$\vx \geqq \vy$ to denote that $\vx(i) \geq \vy(i)$ for all $1 \leq i \leq k$,
$\vx \geq \vy$ to denote that $\vx \geqq \vy$ but $\vx \neq \vy$,
and $\vx > \vy$ to denote that $\vx(i) > \vy(i)$ for all $1 \leq i \leq k$.
In the case that $\vy = \vec{0}$,
we say that $\vx$ is 
\emph{nonnegative},
\emph{semipositive},
and \emph{positive},
respectively.
Similarly define $\leqq,\leq,<$.

For a matrix or vector $\vx$,
define $\| \vx \| = \| \vx \|_1 = \sum_{i} |\vx(i)|$,
$i$ ranges over all the entries of $\vx$.

\subsection{Chemical Reaction Networks}

A \emph{chemical reaction network} (CRN) is a pair $\C=(\Lambda, R)$, where $\Lambda$ is a finite set of chemical \emph{species}, and $R$ is a finite set of reactions over $\Lambda$, where each \emph{reaction} is a pair $(\vr,\vp) \in \N^\Lambda \times \N^\Lambda$ indicating the \emph{reactants} $\vr$ and \emph{products} $\vp$. 
A \emph{population protocol}~\cite{angluin2004computation} is a CRN in which all reactions $(\vr,\vp)$ obey $\| \vr \| = \| \vp \| = 2$.
(Note that CRNs, including population protocols, do not assume any underlying ``communication graph'' and model a well-mixed system in which each equal-sized of molecules is as likely to collide and react as any other.)
We write reactions such as $A+2B \rxn A+3C$ to represent the reaction $(\{A,2B\}, \{A,3C\})$.
A \emph{configuration} $\vec{c} \in \N^{\Lambda}$ of a CRN assigns integer counts to every species $S \in \Lambda$. 
When convenient, we use the notation $\left\{n_1 S_1, n_2 S_2, \ldots, n_k S_k\right\}$ to describe a configuration $\vc$ with $n_i \in \N$ copies of species $S_i$, i.e., $\vc(S_i) = n_i$, 
and any species that is not listed is assumed to have a zero count. 
If some configuration $\vc$ is understood from context, for a species $S$, we write $\# S$ to denote $\vc(S).$
A reaction $(\vr, \vp)$ is said to be \emph{applicable} in configuration $\vc$ if $\vr \leqq \vc$. 
If the reaction $(\vr, \vp)$ is applicable,
applying it results in configuration $\vc' = \vc - \vr + \vp$, and we write $\vc \rightarrow \vc'$.

An \emph{execution} $\E$ is a finite or infinite sequence of one or more configurations $\E=(\vec{c}_0, \vec{c}_1, \vec{c}_2, \ldots)$ such that, for all $i \in\{1, \ldots,|\E|-1\}, \vec{c}_{i-1} \rightarrow \vec{c}_i$ and $\vec{c}_{i-1} \neq \vec{c}_i$.
$\vec{x} \reach_P \vec{y}$ denotes that $P$ is finite, starts at $\vec{x}$, and ends at $\vec{y}$.
In this case we say $\vy$ is \emph{reachable} from $\vx$.
Let $\reachset{\vx} = \{ \vy \mid \vx \reach \vy\}$.
Note that the reachability relation is \emph{additive}:
if $\vx \reach \vy$, then for all $\vc \in \N^\Lambda$,
$\vx + \vc \reach \vy + \vc$.

For a CRN $\calC=(\Lambda,R)$ where $|\Lambda|=n$ and $|R|=m$,
define the $n \times m$ \emph{stoichiometric matrix} $\vM$ of $\calC$ as follows.
The species are ordered $S_1,\dots,S_n$, and the reactions are ordered $(\vr_1,\vp_1),\dots,(\vr_m,\vp_m)$,
and $\vM_{ij} = \vp_j(S_i) - \vr_j(S_i)$.
In other words, $\vM_{ij}$ is the net amount of $S_i$ produced when executing the $j$'th reaction.
For instance, if the CRN has two reactions $S_1 \rxn S_2 + 2S_3$ and $3S_2 + S_3 \rxn S_1 + S_2 + S_3$,
then 
\opt{full}{
\[
    \vM = 
\begin{pmatrix}
-1 & 1 \\
1  & -2 \\
2  & 0
\end{pmatrix}.
\]
}
\opt{submission,final}{
$
    \vM = 
\begin{pmatrix}
-1 & 1 \\
1  & -2 \\
2  & 0
\end{pmatrix}.
$
}

\begin{remark}
Let $\vu \in \N^R$.
Then the vector $\vM \vu \in \Z^\Lambda$ represents the change in species counts that results from applying reactions by amounts described in $\vu$. 
In the above example, if $\vu = (2,1)$,
then $\vM \vu = (-1,0,4)$,
meaning that executing the first reaction twice ($\vu(1)=2$)
and the second reaction once ($\vu(2)=1$) causes $S_1$ to decrease by 1, $S_2$ to stay the same, and $S_3$ to increase by 4.
\end{remark}

\subsection{Stable computation with CRNs}

To capture the result of computations done by a CRN, we generalize the definitions to include information about how to interpret the final configuration after letting the CRN run until the result cannot change anymore (characterized below as \emph{stable computation}). Computation primarily involves two classes of functions: 1. evaluating predicates $\phi:\N^k \to \{0,1\}$ to determine properties of the input, and 2. executing general functions that map an input configuration to an output, denoted as \( f: \N^k \rightarrow \N \). 

The definitions below reference \emph{input species} $\Sigma \subseteq \Lambda$ and an \emph{initial context} $\vs \in \N^{\Lambda \setminus \Sigma}$.
If $\vs = \vec{0}$ we say that CRN is \emph{leaderless}.
The initial context may be any constant multiset of species, though in practice it tends to be a single ``leader'' molecule.
Furthermore, other initial contexts such as $\{ 2A, 3B \}$ could be produced from a single leader $L$ via a reaction $L \to 2A + 3B$, so we may assume without loss of generality that the initial context, if it is nonzero, is simply a single leader.
In both cases, we say $\vi \in \N^\Lambda$ is a \emph{valid initial configuration} if $\vi = \vs + \vx$, where $\vx(S) = 0$ for all $S \in \Lambda \setminus \Sigma$;
i.e., $\vi$ is the initial context plus only input species.

A \emph{chemical reaction decider} (CRD) is a tuple $\D=(\Lambda, R, \Sigma, \Upsilon_1,\Upsilon_0, \vec{s})$, where $(\Lambda, R)$ is a CRN, $\Sigma \subseteq \Lambda$ is the set of \emph{input species}, $\Upsilon_1 \subseteq \Lambda$ is the set of \emph{yes voters}, and $\Upsilon_0 \subseteq \Lambda$ is the set of \emph{no voters}, such that $\Upsilon_1 \cap \Upsilon_0 = \emptyset$, and $\vec{s} \in \N^{\Lambda \backslash \Sigma}$ is the \emph{initial context}. 
If $\Upsilon_1 \cup \Upsilon_0 = \Lambda$, we say the CRD is \emph{all-voting}.
We define a global output partial function $\Phi: \N^{\Lambda} \dashrightarrow \{0,1\}$ as follows. $\Phi(\vec{c})$ is undefined if either $\vec{c}=\vec{0}$, 
or if there exist $S_0 \in \Upsilon_0$ and $S_1 \in \Upsilon_1$ such that $\vec{c}\left(S_0\right)>0$ and $\vec{c}\left(S_1\right)>0$.
In other words, we require a unanimous vote as our output. 
We say $\vc$ is \emph{stable} if, for all $\vc'$ such that $\vc \reach \vc'$,
$\Phi(\vc) = \Phi(\vc').$
We say a CRD $\mathcal{D}$ \emph{stably decides} the predicate $\psi: \N^{\Sigma} \rightarrow\{0,1\}$ if, 
for any valid initial configuration $\vi \in \N^{\Lambda}$, letting $\vi_0 = \vi \upharpoonright \Sigma$, 
for all configurations $\vc \in \N^{\Lambda}, 
\vi \reach \vc$ implies $\vc \reach \vc^{\prime}$ such that $\vc^{\prime}$ is stable and $\Phi\left(\vc^{\prime}\right)=\psi\left(\vi_0\right)$. 
We associate to a predicate $\psi$ the set $A = \psi^{-1}(1)$ of inputs on which $\psi$ outputs 1,
so we can equivalently say the CRD \emph{stably decides} the set $A.$

A chemical reaction computer $(C R C)$ is a tuple $\C=(\Lambda, R, \Sigma, Y, \vec{s})$, where $(\Lambda, R)$ is a CRN, 
$\Sigma \subset \Lambda$ is the set of \emph{input species}, 
$Y \in \Lambda \backslash \Sigma$ is the \emph{output species}, and $\vec{s} \in \N^{\Lambda \backslash \Sigma}$ is the \emph{initial context}.
A configuration $\vec{o} \in \N^{\Lambda}$ is \emph{stable} if, for every $\vec{c}$ such that $\vec{o} \reach \vec{c}, \vec{o}(Y) = \vec{c}(Y)$, i.e. the output can never change again. 
We say that $\C$ \emph{stably computes} a function $f: \N^k \rightarrow \N$ if for any valid initial configuration $\vec{i} \in \N^{\Sigma}$ and any $\vec{c} \in \N^{\Lambda}, \vec{i} \reach \vec{c}$ implies
$\vec{c} \reach \vec{o}$ such that $\vec{o}$ is stable and $f(\vec{i} \upharpoonright \Sigma)=\vec{o}(Y)$.

\subsection{Time model}

The following model of stochastic chemical kinetics is widely used in quantitative biology and other fields dealing with chemical reactions between species present in small counts~\cite{Gillespie77}.
It ascribes probabilities to execution sequences, and also defines the time of reactions, allowing us to study the computational complexity of the CRN computation in \Cref{sec:positive-results-sets,sec:positive-results-fuctions}.
If the volume is defined to be the total number of molecules,
then the time model is essentially equivalent to the notion of \emph{parallel time} studied in population protocols~\cite{AngluinAE2008Fast}.
In this paper, the rate constants of all reactions are $1$, and we define the kinetic model with this assumption.
A reaction is \emph{unimolecular} if it has one reactant and \emph{bimolecular} if it has two reactants.
We use no higher-order reactions in this paper.

The kinetics of a CRN is described by a continuous-time Markov process as follows.
Given a fixed volume $v > 0$,
the \emph{propensity} of a unimolecular reaction $\alpha : X \to \ldots$ in configuration $\vc$ is $\rho(\vc, \alpha) = \vc(X)$.
The propensity of a bimolecular reaction $\alpha : X + Y \to \ldots$, where $X \neq Y$, is $\rho(\vc, \alpha) = \frac{\vc(X) \vc(Y)}{v}$.
The propensity of a bimolecular reaction $\alpha : X + X \to \ldots$ is $\rho(\vc, \alpha) = \frac{1}{2} \frac{\vc(X) (\vc(X) - 1)}{v}$.
The propensity function determines the evolution of the system as follows.
The time until the next reaction occurs is an exponential random variable with rate $\rho(\vc) = \sum_{\alpha \in R} \rho(\vc,\alpha)$ (note that $\rho(\vc)=0$ if no reactions are applicable to $\vc$).
The probability that next reaction will be a particular $\alpha_{\text{next}}$ is $\frac{\rho(\vc,\alpha_{\text{next}})}{\rho(\vc)}$.

The kinetic model is based on the physical assumption of well-mixedness that is valid in a dilute solution.
Thus, we assume the \emph{finite density constraint}, which stipulates that a volume required to execute a CRN must be proportional to the maximum molecular count obtained during execution~\cite{SolCooWinBru08}.
In other words, the total concentration (molecular count per volume) is bounded.
This realistically constrains the speed of the computation achievable by CRNs.

For a CRD or CRC stably computing a predicate/function,
the \emph{stabilization time} is the function $t:\N \to \N$ defined for all $n \in \N$ as $t(n)=$ the worst-case expected time to reach from any valid initial configuration of size $n$ to a stable configuration.

\subsection{Semilinear sets, predicates, functions}

\begin{definition}
    A set $L \subseteq \N^d$ is \emph{linear} if there are vectors $\vb,\vp_1,\dots,\vp_k$ such that
    $L = \{ \vb + n_1 \vp_1 + \dots + n_k \vp_k \mid n_1,\dots,n_k \in \N \}$.
    A set is \emph{semilinear} if it is a finite union of linear sets.
    A predicate $\phi:\N^d \to \{0,1\}$ is \emph{semilinear} if the set $\phi^{-1}(1)$ is semilinear.
    A function $f: \N^d \to \N$ is \emph{semilinear} if its \emph{graph}
    $\{ (\vx,y) \in \N^{d+1} \mid f(\vx) = y \}$
    is semilinear.
\end{definition}

The following is a known characterization of the computational power of CRNs~\cite{angluin2007computational,chen2023rate}.

\begin{theorem}
[\cite{angluin2007computational,chen2023rate}]
    A predicate/function is stably computable by a CRD/CRC if and only if it is semilinear.
\end{theorem}

\begin{definition}
    $T \subseteq \N^d$ is a \emph{threshold} set is if there are constants $c,w_1,\dots,w_d \in \Z$ such that
    $T = \{ \vx \in \N^d \mid w_1 \vx(1) + \dots + w_d \vx(d) \leq c \}.$
    $M \subseteq \N^d$ is a \emph{mod} set if there are constants $c,m,w_1,\dots,w_d \in \N$ such that
    $M = \{ \vx \in \N^d \mid w_1 \vx(1) + \dots + w_d \vx(d) \equiv c \mod m \}.$
    \label{defn:threshold_and_mod_set}
\end{definition}

The following well-known characterization of semilinear sets is useful.

\begin{theorem}
[\cite{Ginsburg1966Semigroups}]
\label{thm:semilinear-Boolean-combination-threshold-mod}
    A set is \emph{semilinear} if and only if it is a Boolean combination (union, intersection, complement) of threshold and mod sets.
\end{theorem}


\section{Execution bounded chemical reaction networks}
\label{sec:execution-bounded-crns}

In this section, we define execution bounded CRNs and state 
\opt{submission,final}{an alternate characterization}
\opt{full}{some alternate characterizations}
of the definition.


\begin{definition}
    A CRN $\C$ is \emph{execution bounded from configuration $\vec{x}$} if all executions $\E = (\vec{x}, \ldots )$ starting at $\vec{x}$ are finite. 
    A CRD or CRC $\calC$ is \emph{execution bounded} if it is execution bounded from every valid initial configuration.
    $\calC$ is \emph{entirely execution bounded} if it is execution bounded from every configuration.
    \label{def:execution-bounded}
\end{definition}

This is a distinct concept from the notion of ``bounded'' CRNs studied by Rackoff~\cite{Rackoff1978CoveringBoundedness} (studied under the equivalent formalism of vector addition systems).
That paper defines a CRN to be \emph{bounded} from a configuration $\vx$ if $|\reachset{\vx}|$ is finite (and shows that the decision problem of determining whether this is true is $\mathsf{EXPSPACE}$-complete.)
We use the term \emph{execution bounded} to avoid confusion with this concept.

\begin{toappendix}

We first observe that being execution bounded from $\vx$ implies a slightly stronger condition:
there is a uniform upper bound on the length of \emph{all} executions from $\vec{x}$.\footnote{
    In other words, this rules out the possibility that, although all executions from $\vec{x}$ are finite, there are infinitely many of them $\E_1,\E_2,\dots$, each finite but longer than the previous.
}

\begin{observation}
\label{obs:finite_global_bound_on_execution_length}
    A CRN is execution bounded from $\vec{x}$ if and only if there is a constant $N \in \N$ such that all executions from $\vec{x}$ have length at most $N$.
    Equivalently, there are finitely many executions from $\vec{x}$, all of which are finite.
\end{observation}

\begin{proof}
    We use K\H{o}nig's lemma to show that in the absence of an infinite path, the number of all possible paths must be finite, which directly implies a global bound on the length of \emph{all} executions. We represent the set of all executions for $\C$ as a tree where each edge represents a single reaction applied and each node stores the complete execution sequence starting from configuration \vec{x}. Note that this construction is slightly different from a more straightforward graph with the reachable states as nodes, which would not give us a tree, since the same state can be reached by different executions. Formally, we generate the tree as follows: $T_\C(\vec{x}) = (V, E)$ where $V \triangleq \{ \E \in \{\N^\Lambda \}^* \mid \E \text{ is a valid execution sequence starting from }\vec{x}\}$, $E \triangleq \{ (\E_1, \E_2 ) \mid \E_1 \preceq \E_2 \land |\E_2| = |\E_1| + 1$\}.
    In other words, all the executions from $\vec{x}$ of length $d$ are the nodes at depth $d$ of this tree.
    One can think of the nodes as being labeled by configurations rather than executions (specifically the final configuration of the execution, with the tree rooted at $\vec{x}$), but the same configuration can label multiple nodes if it can be reached from $\vec{x}$ via different executions.
    In this case the children of a configuration are those that are reachable from it by applying a single reaction.
    
    This tree is finitely branching, as we can only choose from a finite number of reactions at any node. By definition of execution bounded, there is no execution sequence with an infinite length.
    Due to the bijection between paths in $T_\C(\vec{x})$ and executions possible in $\C$, there is no infinite path in the tree.
    By K\H{o}nig's Lemma, the tree has a finite number of nodes, guaranteeing a single bound $N$ (the depth of the tree) on the length of every execution.
\end{proof} 

The next lemma characterizes execution boundedness as equivalent to having a finite reachable state space with no cycles.

\begin{lemma}
     A CRN is execution bounded from $\vec{x}$ if and only if $\reachset{\vx} = \{ \vec{y} \mid \vec{x} \reach \vec{y} \}$ is finite and, for all $\vy \in \reachset{\vx}$, $\vec{y} \not \reach \vec{y}$ except by the zero-length execution.
\end{lemma}

\begin{proof}
    Every configuration reachable from \vec{x} is reached through some execution contained in $T_\C(\vec{x})$ as a node and there exists only a finite number of them (\Cref{obs:finite_global_bound_on_execution_length}). Multiple unique executions can produce the same configuration but one execution cannot produce multiple configurations. 
    Thus, there exists a surjection from the nodes of $T_\C(\vec{x})$ into $\reachset{\vx}$ and $\reachset{\vx}$ must also be finite. For the second part of the condition, we prove its contrapositive and assume there exists $\vec{y} \in \reachset{\vx}$, $\vec{y} \reach_P \vec{y} \land |P| > 0$. Let $P = (\vec{p}_1, \vec{p}_2, \dots, \vec{p}_n)$. It holds that $\vec{p}_1 = \vec{p}_n$ and $\vec{p}_{n-1} \reach \vec{p}_1$. We can construct an infinite-length execution $P^\prime = (\vec{x}, \dots, \vec{p}_1,\dots, \vec{p}_{n-1}, \vec{p}_1, \dots)$, which must also be a valid under the reactions of $\C$, making $\C$ execution unbounded from $\vec{x}$.

    If $\reachset{\vx}$ is finite and contains no such $\vec{y}$, then we can construct a finite, directed, acyclic graph $G_\C(\vec{x}) = (V, E)$ where $V=\reachset{\vx}$, $E=\{(\vec{x}, \vec{y}) \mid \vec{x} \reach_P \vec{y} \land |P| > 0\}$. The longest path in the graph has length of at most $|\reachset{\vx}|-1$. A bijection exists between paths in $G_\C(\vec{x})$ and executions possible in $\C$ starting from \vec{x}. We set $n = |\reachset{\vx}|$ satisfying that each execution has length of at most $n$, making $\C$ execution bounded.
\end{proof}

The following result is used frequently in impossibility proofs for CRNs and population protocols,
and it will help us prove another characterization of execution bounded CRNs in \Cref{lem:execution-bdd-iff-no-self-covering-path}.

\begin{lemma}
\emph{(Dickson's Lemma)}
    For every infinite sequence of nonnegative integer vectors $\vx_1,\vx_2,\dots \in \N^k$,
    there are $i < j$ such that
    $\vx_i \leqq \vx_j$.
\end{lemma}
\end{toappendix}

We first observe an equivalent characterization of execution bounded that will be useful in the negative results of \Cref{sec:limitations}.
\opt{submission}{A proof is in the appendix.}

\begin{definition}
    A execution $\calE = (\vx_1,\vx_2,\dots)$ is \emph{self-covering} if for some $i < j$,
    $\vx_i \leqq \vx_j$.
    It is \emph{strictly self-covering} if $\vx_i \leq \vx_j$.
    We also refer to these as (strict) self-covering \emph{paths}.\footnote{
    Rackoff~\cite{Rackoff1978CoveringBoundedness} uses the term ``self-covering'' to mean what we call \emph{strictly self-covering} here,
    and points out that Karp and Miller~\cite{karp1969parallel} showed that $|\reachset{\vx}|$ is infinite if and only if there is a strictly self-covering path from $\vx$.
    The distinction between these concepts is illustrated by the CRN $A \revrxn B$.
    From any configuration $\vx$, $\reachset{\vx}$ is finite
    ($|\reachset{\vx}| = \vx(A) + \vx(B) + 1$), and there is no strict self-covering path.
    However, from (say) $\{A\}$,
    there is a (nonstrict) self-covering path $\{A\} \reach \{B\} \reach \{A\}$, and by repeating, this CRN has an infinite cycling execution within its finite configuration space $\mathsf{reach}(\{A\}) = \{\{A\},\{B\}\}$.
    }
\end{definition}

\begin{lemmarep}
\label{lem:execution-bdd-iff-no-self-covering-path}
    A CRN is execution bounded from $\vx$ if and only if there is no self-covering path from $\vx$.
\end{lemmarep}

\begin{proof}
    For (the contrapositive of) the forward direction, assume there is a self-covering path from $\vx$, which reaches to $\vx_i$ and later to $\vx_j \geqq \vx_i$.
    By additivity the reactions leading from $\vx_i$ to $\vx_j$ can be repeated indefinitely (in a cycle if $\vx_i=\vx_j$, and increasing some molecular counts unboundedly if $\vx_i \leq \vx_j$),
    so $\calC$ is not execution bounded from $\vx$.

    For the reverse direction,
    assume $\calC$ is not execution bounded from $\vx$.
    Then there is an infinite execution $\calE = (\vx=\vx_1,\vx_2,\vx_3,\dots).$
    By Dickson's Lemma there are $i<j$ such that $\vx_i \leqq \vx_j$,
    i.e., $\calE$ is self-covering.
\end{proof}

\section{Execution bounded CRDs stably decide all semilinear sets}
\label{sec:positive-results-sets}

In this section, we will show that execution bounded CRDs have the same computational power as unrestricted CRDs.
The following is the main result of this section.

\begin{theorem}
\label{thm:semilinear_sets_decidable}
    Exactly the semilinear sets are stably decidable by execution bounded CRDs.
    Furthermore, each can be stably decided with expected stabilization time $\Theta(n \log n)$.
\end{theorem}

Since semilinear sets are Boolean combinations of mod and threshold predicates, we prove this theorem by showing that execution bounded CRDs can decide mod and threshold sets individually as well as any Boolean combination in the following lemmas. To ensure execution boundedness in the last step, we require the following property.

\begin{definition}
    Let $\calD$ be a CRD with voting species $\Upsilon$.
    We say $\calD$ is \emph{single-voting} if for any valid initial configuration $\vi \in \N^\Sigma$ and any $\vc \in \N^\Lambda$ s.t. $\vi \reach \vc$, 
    $\sum_{V \in \Upsilon} \vc(V) = 1$,
    i.e., exactly one voter is present in every reachable configuration.
\end{definition}



\opt{submission,final}{
\Cref{lem:mod_sets_decidable,lem:threshold_sets_decidable} are proven in
\opt{submission}{the appendix}.
\opt{final}{the full version of this paper.}
}

\begin{lemmarep}
\label{lem:mod_sets_decidable}
    Every mod set $M=\big\{\left(x_1, \ldots, x_d\right) \mid \sum_{i=1}^d w_i x_i \equiv c \bmod m\big\}$ is stably decidable by an execution bounded, single-voting CRD with expected stabilization time $\Theta(n\log n)$.
\end{lemmarep}

We design a CRD $\D$ with exactly one leader present at all times, cycling through $m$ ``states'' while consuming the input and accepting on state $c$.
Let $\Sigma = \{X_1, \dots, X_d\}$ be the set of input species and start with only one $L_0$ leader, i.e. set the initial context $\vec{s}(L_0)=1$ and $\vec{s}(S)=0$ for all other species. For each $i \in \{1, \dots, d\}, j \in \{0, \dots ,m-1\}$ add the following reaction:
$
    X_i + L_j \rightarrow L_{j + w_i \bmod m}.
$
Let only $L_c$ vote \emph{yes} and all other species \emph{no}, i.e. $\Upsilon = \{L_c\}$. For any valid initial configuration, $\D$ reaches a stable configuration which votes \emph{yes} if and only if the input is in the mod set, and \emph{no} otherwise. \opt{submission}{The  time and execution boundedness are proven in the appendix.}

\begin{proof}
$\calD$ terminates with the correct output value: At any point in time, there is a single leader $L_j$ present (the initial configuration contains a single leader and each reaction produces and consumes one). Every reaction satisfies the following invariant (for the leader's subscript $j$): $j \equiv \sum_{i=1}^d w_i x_i' \bmod m$ where $x_i'$ is the updated count of species $X_i$ in the current configuration. By design of $\D$, there will be a reaction applicable as long as there are copies of $X_i$ (a leader with any subscript can react with any $X_i$). After applying this reaction as often as possible, we have reached a stable configuration with $L_{\sum_{i=1}^{d} w_i x_i \bmod m}$ as the only species present.

$\calD$ is execution bounded: Every reaction reduces the count of chemicals by one. Every possible execution contains exactly $\|\vec{i}\|$ configurations, where $\|\vec{i}\|$ is the number of all molecules in the starting configuration.

$\calD$ is single-voting: Initially, $L_0$ is present and the only voter. Every valid input contains no voter and every reaction results in no change to the count of copies of $L_i$. 

$\calD$ stabilizes in $\Theta(n \log n)$ time: We start with $\#L = 1, \#X = n$ in volume $n$. $n$ reactions must occur before $\D$ terminates. For the first reaction, we have a rate of $\lambda=\frac{n \cdot 1}{n}$, for the last (with only the leader and one $X$ present), our rate will be $\lambda = \frac{1 \cdot 1}{n}$. Thus, the expected time for all $n$ reactions to complete is 
\[
\sum_{i=1}^{n} \frac{n}{i} = n \sum_{i=1}^{n} \frac{1}{i} = \Theta(n \log n).
\qedhere
\]
\end{proof}

\begin{lemmarep}
\label{lem:threshold_sets_decidable}
    Every threshold set $T=\big\{(x_1, \ldots, x_d) \mid \sum_{i=1}^d w_i x_i \geq t\big\}$ is stably decidable by an execution bounded, single-voting CRD with expected stabilization time $\Theta(n \log n)$.
\end{lemmarep}

We design a CRD $\D$ which multiplies the input molecules according to their weight and consumes positive and negative units alternatingly using a single leader. Once no more reaction is applicable, the leader's state will indicate whether or not there are positive units left and the threshold is met. Let $\Sigma = \{X_1, \dots, X_d\}$ be the set of input species and $\Upsilon = \{L_Y\}$ the \emph{yes} voter. We first add reactions to multiply the input species by their respective weights. For all $i \in \{1, \dots, d\}$, add the reaction:
\begin{align}
    X_i \rightarrow \begin{cases}
        w_i P &\text{if } w_i>0 \\
        -w_i N &\text{if } w_i<0 \\
        \emptyset &\text{otherwise}
    \end{cases}
    \label{rxn:x_to_weighted_pos_neg}
\end{align}
$P$ and $N$ represent ``positive'' and ``negative'' units respectively. Now add reactions to consume $P$ and $N$ alternatingly using a leader until we run out of one species:
\begin{align}
    L_Y + N &\rightarrow L_N \label{rxn:yes_to_no}\\
    L_N + P &\rightarrow L_Y \label{rxn:no_to_yes}
\end{align}
Finally, initialize the CRD with one $L_Y$ and the threshold number $t$ copies of $P$ (or $-tN$ if $t$ is negative), i.e. $\vec{s}(L_Y)=1$, $\vec{s}(P)=t$ if $t>0$, or $\vec{s}(N)=-t$ if $t<0$, and $\vec{s}(S)=0$ for all other species. For any valid initial configuration, $\D$ reaches a stable configuration which votes \emph{yes} if and only if the weighted sum of inputs is above the threshold, and \emph{no} otherwise. 
\opt{submission}{The execution time is proven in the appendix.}

\begin{proof}
    $\calD$ is single-voting since it starts with a single leader and no reaction changes the count of $L_B$ molecules.
    
    $\calD$ stabilizes in $\Theta(n \log n)$ time: First, all input species will be converted to $w_i$ instances of $P$ or $N$.  We run these reactions until no $X_i$. As they are independent of molecules other than the reactant, these reactions have a rate of $\lambda = i$, so the expected time until the next reaction is $\frac{1}{i}$. The total time for reactions \eqref{rxn:x_to_weighted_pos_neg} to complete is therefore $\sum_{i=1}^n \frac{1}{i} = \mathit{\Theta}(\log n)$. The time for reactions \eqref{rxn:yes_to_no} and \eqref{rxn:no_to_yes} on the other hand is asymptotically dominated by the last reaction, where $\#L=1$ and $\#B=1$, where $B \in \{P, N\}$, so $\lambda = \frac{1\cdot1}{n}$. Let $n_P, n_N$ be the counts of $P, N$ and assume without loss of generality $n_P \geq n_N$. We get:
    \[
    \sum_{i=0}^{n_N-1} \frac{n}{(n_P-i)} + \frac{n}{(n_N-i)} \leq
    \sum_{i=0}^{n_N-1} \frac{2n}{(n_N-i)} = 
    2n \sum_{i=1}^{n_N} \frac{1}{i} = 
    \Theta(n \log n).
    \qedhere
    \]
\end{proof}

\begin{lemma}
\label{lem:boolean-closure-crd}
    If sets $X_1, X_2 \subseteq \N^d$ are stably decided by some execution bounded, single-voting CRD, then so are $X_1 \cup X_2, X_1 \cap X_2$, and $\overline{X_1}$ with expected stabilization time $O(n \log n)$.
\end{lemma}

\begin{proof}
To stably decide $\overline{X_1}$, swap the yes and no voters.

For $\cup$ and $\cap$, consider a construction where we decide both sets separately and record both of their votes in a new voter species. For this, we allow the set of all voters to be a strict subset of all species. We first add reactions to duplicate our input with reactions of the form \begin{equation}X_i \rightarrow X_{i,1}+X_{i,2}\label{rxn:split_input_for_parallel_crns}\end{equation} by two separate CRDs. Subsequently, we add reactions to record the separate votes in one of four new voter species: $V_{N N}, V_{N Y}, V_{Y N}, V_{Y Y}$. The first and second CRN determine the first and second subscript respectively. For $b \in \{Y, N\}$ and if $S_b, T_b$ are voters of $\C_1$ and $\C_2$ respectively, add the reactions:
\begin{align}
    S_b+V_{\overline{b} ?} \rightarrow S_b+V_{b ?}
    \label{rxn:boolean-closure-crd-flip-first-vote}
    \\ 
    T_b+V_{? \overline{b}} \rightarrow T_b+V_{? b}
    \label{rxn:boolean-closure-crd-flip-second-vote}
\end{align} 
Above, the ? subscript is shorthand for ``any bit''; e.g. if $N_1$ is the \emph{no} voter of the first CRD, we would add two reactions $N_1+L_{YN} \rightarrow N_1 + L_{NN}$ and $N_1+L_{YY} \rightarrow N_1+L_{N Y}$. We let the \emph{yes} voters be: $\Upsilon = \{V_{N Y}, V_{Y N}, V_{Y Y}\}$ to stably decide $X_1 \cup X_2$ or $\Upsilon = \{V_{YY}\}$ to stably decide $X_1 \cap X_2$.

    Reaction \eqref{rxn:split_input_for_parallel_crns} will complete in $O(\log n)$ time and is clearly execution bounded since the input $X_i$ is finite and not produced in any reaction. Consequently, two separate CRNs run in $\mathit{\Theta}(n \log n)$ time as shown in \cref{lem:mod_sets_decidable} and \cref{lem:threshold_sets_decidable}. After stabilization of the parallel CRNs, we expect reaction \eqref{rxn:boolean-closure-crd-flip-first-vote} and \eqref{rxn:boolean-closure-crd-flip-second-vote} to happen exactly once. Each molecule involved is a leader and has count $1$ in volume $n$. This leads to a rate of $\lambda=\frac{1\cdot1}{n}$, so the expected time for one reaction to happen is $O(n)$. It is important to note that reactions \eqref{rxn:boolean-closure-crd-flip-first-vote} and \eqref{rxn:boolean-closure-crd-flip-second-vote} do not result in unbounded executions due to the unanimous vote in parallel CRDs. In both mod sets and threshold sets, the leader changes its vote a maximum of $|\vec{i}|$ times, with only ever one leader present at any time. Again, we start with only one $V_{bb}$ voter present initially and no reaction changes the count of voters, making our construction single-voting.
\end{proof}

Since semilinear predicates are exactly Boolean combinations of threshold and mod predicates, 
\Cref{lem:mod_sets_decidable,lem:threshold_sets_decidable,lem:boolean-closure-crd} imply \Cref{thm:semilinear_sets_decidable}.

We can also prove the same result for all-voting CRDs.
Note, however, that such CRDs cannot be ``composed'' using the constructions of \Cref{lem:boolean-closure-crd,thm:execution-bdd-CRCs-compute-all-semilinear-functions},
which crucially relied on the assumption that the CRDs being used as ``subroutines'' are single-voting.
\opt{submission}{A proof is in the appendix.}

\begin{thmrep}
\label{lem:boolean-closure-crd-all-voting}
    Every semilinear set is stably decidable by an execution bounded, all-voting CRD, with expected stabilization time $O(n \log n).$
\end{thmrep}

\begin{proof}
    By \Cref{thm:semilinear_sets_decidable},
    every semilinear set is stably decided by  a single-voting CRD.
    We convert this to an all-voting CRD, where every species is required to vote \emph{yes} or \emph{no}, by ``propagating'' the final vote (recorded in the single voter $V^0$ voting \emph{no} or $V^1$ voting \emph{yes}) back to all other molecules. 
    A superscript indicates the ``global'' decision. 
    The execution boundedness proven in \cref{lem:boolean-closure-crd} ensures that the leader propagates the final vote only a finite amount of times. 
    For each vote $b \in \{0,1\}$ and each voter $V^b$ voting $b$,
    and all other species $S \in \Lambda \backslash \{V\}$, replace species $S$ with two versions $S^0$ and $S^1$, and add reactions:
    \begin{align}
        V^b+S^{\overline{b}} \rxn V^b+S^b
        \label{rxn:leader_propagate_global_vote}
    \end{align}

    The original reactions of the CRD must also be replaced with ``functionally identical'' reactions for the new versions of species. For example, the reaction $A+B \rxn C+D$ becomes
    \begin{align*}
        A^0+B^0 &\rxn C^0+D^0
        \\
        A^0+B^1 &\rxn C^0+D^0
        \\
        A^1+B^0 &\rxn C^0+D^0
        \\
        A^1+B^1 &\rxn C^1+D^1
    \end{align*}
    In the middle two cases we can pick the superscripts of the products arbitrarily, 
    whereas in the first and last case, we must choose the product votes to match those of the reactants to ensure stable states remain stable.
    
    A vote change of the $V^b$ leader leads to the propagation of the vote to at most $n$ molecules once using reaction \eqref{rxn:leader_propagate_global_vote}.
    This reaction dominates the runtime, as a single molecule is required to interact with each other molecule. We cannot speed this process up using an epidemic style process as conflicting votes would make the CRN execution unbounded. 
    The original CRD takes time $O(n \log n)$ to converge on a correct output for the single voter $V^b$.
    At that point, a standard coupon collector argument shows that the voter $V^b$ takes expected time $O(n \log n)$ to correct the votes of all other species via reaction \eqref{rxn:leader_propagate_global_vote}.
\end{proof}

\section{Execution bounded CRCs stably compute all semilinear functions}
\label{sec:positive-results-fuctions}

In this section we shift focus from computing Boolean-valued predicates $\phi: \N^d \to \{0,1\}$ to integer-valued functions $f: \N^d \to \N$,
showing that execution bounded CRCs can stably compute the same class of functions (semilinear) as unrestricted CRCs.


Similar to \cite{Chen2012DeterministicFunction,Doty2013Leaderless}, we compute semilinear functions by decomposing them into ``affine pieces'', which we will show can be computed by execution bounded CRNs and combined by using semilinear predicates to decide which linear function to apply for a given input.\footnote{
    While this proof generalizes to multivariate output functions as in \cite{Chen2012DeterministicFunction,Doty2013Leaderless}, to simplify notation we focus on single output functions.
    Multi-valued functions $f: \N^d \to \N^l$ can be equivalently thought of as $l$ separate single output functions $f_i: \N^d \to \N$,
    which can be computed in parallel by independent CRCs.
}

We say a partial function $f: \N^k \dashrightarrow \N$ is \emph{affine} if there exist a vectors $\vec{a} \in \Q^k$, $\vec{c} \in \N^k$ with $\vec{x}-\vec{c} \geq \vec{0}$ and nonnegative integer $b \in \N$ such that 
$
f(\vec{x}) = \vec{a}^\top(\vec{x} - \vec{c})+b.
$
For a partial function $f$ we write $\operatorname{dom} f$ for the \emph{domain} of $f$, 
the set of inputs for which $f$ is defined.
This definition of affine function may appear contrived, but the main utility of the definition is that it satisfies \cref{lem:partial_affine_function_decomposition_disjoint_domains}.
For convenience, we can ensure to only work with integer valued molecule counts by multiplying by $\frac{1}{d}$ after the dot product, where $d$ may be taken to be the least common multiple of the denominators of the rational coefficients in the original definition such that $n_i = d \cdot \va(i)$:
\opt{full}{
\[
    f(\vec{x}) = b+\sum_{i=1}^k a_i (x_i-c_i) 
    \iff f(\vec{x}) = b+\frac{1}{d} \sum_{i=1}^k n_i (x_i-c_i).
\]
}
\opt{submission,final}{
$
    f(\vec{x}) = b+\sum_{i=1}^k \va(i) (\vx(i)-\vc(i)) 
    \iff 
    f(\vec{x}) = b+\frac{1}{d} \sum_{i=1}^k n_i (\vx(i)-\vc(i)).
$
}

We say that a partial function $\hat f: \N^k \rightarrow \N^2$ is a \emph{diff-representation} of $f$ if $\operatorname{dom} f=\operatorname{dom} \hat{f}$ and, for all $\vec{x} \in \operatorname{dom} f$, if $\left(y_P, y_C\right)=\hat{f}(\vec{x})$, then $f(\vec{x})=y_P-y_C$, and $y_P=O(f(\vec{x}))$. In other words, $\hat{f}$ represents $f$ as the difference of its two outputs $y_P$ and $y_C$, with the larger output $y_P$ possibly being larger than the original function's output, but at most a multiplicative constant larger \cite{Doty2013Leaderless}.

\begin{lemma}
\label{lem:partial_affine_function_computable}
    Let $f: \N^k \rightarrow \N$ be an affine partial function. Then there is a diff-representation $\hat{f}: \N^k \longrightarrow \N^2$ of $f$ and an execution bounded CRC that monotonically stably computes $\hat{f}$ in expected stabilization time $O(n)$.
\end{lemma}

\begin{proof}
Define a CRC $C$ with input species $\Sigma=\{X_1, \ldots, X_k\}$ and output species $\Gamma=\{Y^P, Y^C\}$. 
We need to ensure that after stabilizing, $y=\#Y^P - \#Y^C$

To account for the $b$ offset, start with $b$ copies of $Y^P$.

For the $c_i$ offset, we must reduce the number of $X_i$ by $c_i$. Since the result will be used in the next reaction, we want to produce a new species $X_i'$ and require $X_i'$ to not be consumed during the computation. We achieve this by adding reactions that let $X_i$ consume itself $c_i$ times (keeping track with a subscript) and converting $X_i$ to $X_i'$ once $c_i$ has been reached.
For the sake of notation below, assume input species $X_i$ is actually named $X_{i,1}$.
For each $i \in\{1, \ldots, k\}$ and $m, p \in\left\{1, \ldots, c_i\right\}$, if $m+p \leq c_i$, add the reaction

\begin{equation}
X_{i, m}+X_{i, p} \rightarrow X_{i, m+p} \label{rxn:accumulate_subscripts_for_subtraction}
\end{equation}
If $m+p>c_i$, add the reaction
\begin{align}
    X_{i, m}+X_{i, p} \rightarrow X_{i, c_i}+\left(m+p-c_i\right) X_i^{\prime}
    \label{rxn:reduce_subscript_for_subtraction}
\end{align}
Runtime: In volume $n$, the rate of reactions \eqref{rxn:accumulate_subscripts_for_subtraction} and \eqref{rxn:reduce_subscript_for_subtraction} would be $\lambda \approx \frac{(x_i)^2}{n}$ ($x_i$ molecules have the chance to react with any of the $x_i-1$ others), so the expected time for the next reaction is $\frac{n}{(x_i)^2}$. The expected time for the whole process is $\sum_{i=1}^{x_i} \frac{n}{i^2}=n \sum_{i=1}^{x_i} \frac{1}{i^2}=O(n)$. Further, the reactions are execution bounded since both strictly decrease the number of their reactants and exactly $x_i - 1$ reactions will happen.

To account for the $n_i / d$ coefficient, we multiply by $n_i$, then divide by $d$ using similar reactions as for the subtraction.
To multiply by $n_i$, add the following reaction for each $i \in\{1, \ldots, k\}$:
\begin{align}
X_i' \rightarrow \begin{cases}n_i D_1^P, & \text { if } n_i>0 \\ \left(-n_i\right) D_1^C, & \text { if } n_i<0\end{cases}
\label{rxn:multiply_by_n}
\end{align}
For each $m, p \in\left\{1, \ldots, d-1\right\}$, if $m+p \leq d-1$, add the reactions
\begin{align}
D_m^P+D_p^P \rightarrow D_{m+p}^P \label{rxn:accumulate_subscripts_for_division_pos}\\
D_m^C+D_p^C \rightarrow D_{m+p}^C
\label{rxn:accumulate_subscripts_for_division_neg}
\end{align}
If $m+p>c_i$, add the reactions
\begin{align}
D_m^P+D_p^P \rightarrow D_{m+p-d}^B+Y^P \label{rxn:reduce_subscripts_for_division_pos} \\
D_m^C+D_p^C \rightarrow D_{m+p-d}^B+Y^C \label{rxn:reduce_subscripts_for_division_neg}
\end{align}
Reactions \eqref{rxn:multiply_by_n} complete in expected time $O(\log n)$, while \eqref{rxn:reduce_subscripts_for_division_pos} and \eqref{rxn:reduce_subscripts_for_division_neg} complete in $O(n)$ by a similar analysis as for the first two reactions. 
As for execution boundedness, \eqref{rxn:multiply_by_n} is only applicable once for every $X_i^{\prime}$; all other reactions start with a number of reactants which are a constant factor of $X_i'$ and decrease the count of their reactants by one in each reaction.
\end{proof}

We require the following result due to Chen, Doty, Soloveichik \cite{Chen2012DeterministicFunction}, guaranteeing that any semilinear function can be built from affine partial functions.
\begin{lemma}[\cite{Chen2012DeterministicFunction}]
\label{lem:partial_affine_function_decomposition_non_disjoint_domains}
    Let $f: \N^d \rightarrow \N$ be a semilinear function. Then there is a finite set $\big\{f_1: \N^d \rightarrow \N, \ldots, f_m: \N^d \rightarrow \N\big\}$ of affine partial functions, where each $\operatorname{dom} f_i$ is a linear set, such that, for each $\vec{x} \in \N^d$, if $f_i(\vec{x})$ is defined, then $f(\vec{x})=f_i(\vec{x})$, and $\bigcup_{i=1}^m \operatorname{dom} f_i=\N^d$.
\end{lemma}

We strengthen \cref{lem:partial_affine_function_decomposition_non_disjoint_domains} to show we may assume each $\mathrm{dom}\ f_i$ is disjoint from the others.
This is needed not only to prove \Cref{thm:execution-bdd-CRCs-compute-all-semilinear-functions},
but to correct the proof of Lemma 4.4 in \cite{Chen2012DeterministicFunction},
which implicitly assumed the domains are disjoint.
\opt{submission}{\Cref{lem:partial_affine_function_decomposition_disjoint_domains} is proven in the appendix.}

\begin{lemmarep}
\label{lem:partial_affine_function_decomposition_disjoint_domains}
    Let $f: \N^d \rightarrow \N$ be a semilinear function. 
    Then there is a finite set $\big\{f_1: \N^d \rightarrow \N, \ldots, f_m: \N^d \rightarrow \N\big\}$ of affine partial functions, where each $\operatorname{dom} f_i$ is a linear set, and $\operatorname{dom} f_i \cap \operatorname{dom} f_j = \emptyset$ for all $i\neq j$, such that, for each $\vec{x} \in \N^d$, if $f_i(\vec{x})$ is defined, then $f(\vec{x})=f_i(\vec{x})$, and $\bigcup_{i=1}^m \operatorname{dom} f_i=\N^d$. 
\end{lemmarep}

\begin{proof}
By \cite[Theorem 2]{Ito1969SemilinearSetsFiniteUnionDisjointLinearSets}, every semilinear set is a finite union of disjoint fundamental linear sets. 
The author defines a linear set $L=\big\{\vec{b}+n_1 \vec{u}_1+\ldots+n_p \vec{u}_p \mid n_1, \ldots, n_p \in \N\big\}$ as \emph{fundamental}, if $\vec{u}_1, \dots \vec{u}_p \in \N^k$ span a $p$-dimensional vector space in $\R^k$, i.e. all vectors are linearly independent in $\R^k$.\footnote{
    This distinction is significant because not all integer-valued linear sets can be represented using solely linearly independent vectors. 
    An illustrative example is $\vec{b}=\vec{0}, \vec{u}_1=(1,1,1), \vec{u}_2=(2,0,1), \vec{u}_3=(0,2,1)$, as discussed in \cite{Chen2012DeterministicFunction}. 
    The vectors $\vec{u}_1, \vec{u}_2, \vec{u}_3$ are not linearly independent in $\R^3$, yet this set cannot be expressed with less than three basis vectors.
}
The proof of \Cref{lem:partial_affine_function_decomposition_non_disjoint_domains} in~\cite{Chen2012DeterministicFunction} shows that each linear set $L_i$ comprising the semilinear graph of $f$ corresponds to one partial affine function $f_i$.
The fact that \cite[Theorem 2]{Ito1969SemilinearSetsFiniteUnionDisjointLinearSets} lets us assume each $L_i$ is disjoint from the others immediately implies that each $\mathrm{dom}\ f_i$ is disjoint from the others.
\end{proof}

The next theorem shows that semilinear functions can be computed by execution bounded CRCs in expected time $O(n \log n)$.

\begin{theorem}
\label{thm:execution-bdd-CRCs-compute-all-semilinear-functions}
    Let $f: \N^d \to \N$ be a semilinear function. 
    Then there is an execution bounded CRC that stably computes $f$ with expected stabilization time $O(n \log n)$.
\end{theorem}

\begin{proof}
    We employ the same construction of \cite{Chen2012DeterministicFunction} with minor alterations. A CRC with input species $\Sigma=\{X_1, \ldots, X_d\}$ and output species $\Gamma=\{Y\}$. By \Cref{lem:partial_affine_function_decomposition_disjoint_domains}, we decompose our semilinear function into partial affine functions (with linear, disjoint domains), which can be computed in parallel by \Cref{lem:partial_affine_function_computable}. Further, we decide which function to use by computing the predicate $\phi_i=$ ``$\mathrm{x} \in \operatorname{dom} f_i$'' (\Cref{thm:semilinear_sets_decidable}). We interpret each $\widehat{Y}_i^P$ and $\widehat{Y}_i^C$ as an ``inactive'' version of ``active'' output species $Y_i^P$ and $Y_i^C$. Let $L_i^Y, L_i^N$ be the \emph{yes} and \emph{no} voters respectively voting whether $\vec{x}$ lies in the domain of $i$-th partial function. Now, we convert the function result of the applicable partial affine function to the global output by adding the following reactions for each $i \in\{1, \ldots, m\}$.
    \begin{align} 
        L_i^Y+\widehat{Y}_i^P & \rightarrow L_i^Y+Y_i^P+Y \label{rxn:local_y_to_global}\\ 
        L_i^N+Y_i^P & \rightarrow L_i^N+M_i \label{rxn:reverse_global_y_1}\\ 
        M_i+Y & \rightarrow \widehat{Y}_i^P \label{rxn:reverse_global_y_2}
    \end{align}
    Reaction \eqref{rxn:local_y_to_global} produces an output copy of species $Y$ and \eqref{rxn:reverse_global_y_1} and \eqref{rxn:reverse_global_y_2} reverse the first reaction using only bimolecular reactions. Both are catalyzed by the vote of the $i$-th predicate result. Also add reactions
    \begin{align} 
        & L_i^Y+\widehat{Y}_i^C \rightarrow L_i^Y+Y_i^C \label{rxn:activate_yc}\\ 
        & L_i^N+Y_i^C \rightarrow L_i^N+\widehat{Y}_i^C \label{rxn:deactivate_yc}
    \end{align}
    and 
    \begin{align} 
        Y_i^P+Y_i^C & \rightarrow K \label{rxn:p_and_c_to_k}\\ 
        K+Y & \rightarrow \varnothing \label{rxn:k_and_y_consumed}
    \end{align}
    Reactions \eqref{rxn:activate_yc} and \eqref{rxn:deactivate_yc} activate and deactivate the ``negative'' output values and reactions \eqref{rxn:p_and_c_to_k} and \eqref{rxn:k_and_y_consumed} allow two active partial outputs to cancel out and consume the excess $Y$ in the process. When the input is in the domain of function $i$, exactly one copy of $L_i^Y$ will be present, otherwise one copy of $L_i^N$. Since we know that the predicate computation is execution bounded and produces at most one voter, the catalytic reaction will also happen at most as often as the leader changes its vote. Therefore, it is also execution bounded.

    The underlying CRNs computing the predicates and functions have expected stabilization time $O(n \log n).$
    Once they have stabilized,
    the slowest reactions described above are those where a leader ($L_i^Y$ or $L_i^N$) must convert all outputs,
    which also takes expected time $O(n \log n)$ by a coupon collector argument.
\end{proof}




\section{Limitations of execution bounded CRNs}
\label{sec:limitations}

The main positive results of the paper (\Cref{thm:semilinear_sets_decidable,thm:execution-bdd-CRCs-compute-all-semilinear-functions})
rely on the assumption that valid initial configurations have a single leader (in particular, they are execution bounded only from configurations with a single leader, but not from arbitrary configurations).
\Cref{lem:boolean-closure-crd-all-voting} shows that we may assume the CRD deciding a semilinear set is all-voting.
However,
for the ``constructive'' results 
\Cref{{lem:boolean-closure-crd},thm:execution-bdd-CRCs-compute-all-semilinear-functions},
which compose the output of a CRD $\calD$ with downstream computation,
using $\calD$ as a ``subroutine'' to stably compute a more complex set/function,
the constructions crucially use the assumption that $\calD$ is single-voting 
(i.e., only the leader of $\calD$ votes)
to argue the resulting composed CRN is execution bounded.
In this section we show these assumptions are necessary,  
proving that execution bounded CRNs without those constraints are severely limited in their computational abilities.

\opt{full,submission}{
We show that entirely execution bounded CRNs (from every configuration) can be characterized by a simpler property of having a ``linear potential function'' that essentially measures how close the CRN is to reaching a terminal configuration.
This result was essentially proven independently by Czerner, Guttenberg, Helfrich, and Esparza~\cite{czerner2024fast} in the model of population protocols.
We give a self-contained proof since we technically need the more general CRN model, although our proof appears to follow their same technique of using a Farkas-like lemma.
}
\opt{final}{We use a result of Czerner, Guttenberg, Helfrich, and Esparza~\cite{czerner2024fast},
showing that entirely execution bounded CRNs (from every configuration) can be characterized by a simpler property of having a ``linear potential function'' that essentially measures how close the CRN is to reaching a terminal configuration.}
We use this characterization to prove that entirely execution bounded CRNs can stably decide only limited semilinear predicates (eventually constant, \Cref{defn:eventually_constant_predicate}),
assuming all species vote, and that molecular counts cannot decrease to $O(1)$ in stable configurations (see \Cref{defn:non-collapsing}).

\subsection{Linear potential functions}

We define a \emph{linear potential function} of a CRN to be a nonnegative linear function of configurations that each reaction strictly decreases.

\begin{definition}
    A \emph{linear potential function} $\Phi: \Rp^\Lambda \to \Rp$ for a CRN is a nonnegative linear function,
    such that for each reaction 
    $(\vr,\vp)$,
    $\Phi(\vp)-\Phi(\vr) < 0.$
\end{definition}

Note that for a configuration $\vx$, since $\Phi(\vx) = \sum_{S \in \Lambda} v_S \vx(S) \geq 0$,
it must be nondecreasing in each species,
i.e., all coefficients $v_S$ must be nonnegative 
(though some are permitted to be 0).
Intuitively, we can think of $\Phi$ as assigning a nonnegative ``mass'' to each species (the mass of $S$ is $v_S$), such that each reaction removes a positive amount of mass from the system.
Note also that since $\Phi$ is linear, the above is equivalent to requiring that $\Phi(\vp - \vr) < 0$,
if we extend $\Phi$ to a linear function $\Phi: \R^\Lambda \to \R$ on vectors with negative elements.


A CRN may or may not have a linear potential function.
Although it is not  straightforward to ``syntactically check'' a CRN to see if  has a linear potential function,
it is efficiently decidable:
a CRN has a linear potential function if and only if the following system of linear inequalities has a solution 
(which can be solved in polynomial time using linear programming techniques; 
the variables to solve for are the $v_S$ for each $S \in \Lambda$),
where the $i$'th reaction has reactants $\vr_i$ and products $\vp_i$,
and species $S \in \Lambda$ has mass $v_S \geq 0$:
\opt{full}{
\[
    (\forall i) 
    \sum_{S \in \Lambda} 
    [ \vp_i(S) - \vr_i(S) ] v_S < 0
\]
}
\opt{submission,final}{
$
    (\forall i) 
    \sum_{S \in \Lambda} 
    [ \vp_i(S) - \vr_i(S) ] v_S < 0.
$
}
For example, for the reactions $A+A \rxn B+C$ and $B+B \rxn A$,
for each reaction to strictly decrease the potential function $\Phi(\vx) = v_A \vx(A) + v_B \vx(B) + v_C \vx(C)$,
$\Phi$ must satisfy 
$2v_A > v_B + v_C$ and $2v_B > v_A$.
In this case, $v_A = 1, v_B = 1, v_C = 0$ works.

\begin{remark}
\label{rmk:linear-potential-integer-coefficients}
A system of linear inequalities with rational coefficients has a real solution if and only if it has a rational solution.
For any homogeneous system (where all inequalities are comparing to 0), any positive scalar multiple of a solution is also a solution.
By clearing denominators, a system has a rational solution if and only if it has an integer solution.
Thus, one can equivalently define a linear potential function to be a function $\Phi(\vx) = \sum_{S \in \Lambda} v_S \vx(S)$
such that each $v_S \in \N$, i.e., we may assume $\Phi:\N^\Lambda \to \N$.
In particular, since $\Phi$ is decreased by each reaction, it is decreased by at least 1.
\end{remark}

\opt{full,submission}{

The following is a variant of Farkas' Lemma~\cite{farkas1902theorie},
one of several similar ``Theorems of the Alternative'' stating that exactly one of two different linear systems has a solution.
(See~\cite[Chapter 2.4]{mangasarian1994nonlinear} for a list of such theorems.)
A proof can be found in~\cite[Theorem 2.10]{gale1960theory}.

\begin{theorem}
    \label{thm:linear-alternative}
    Let $\vM$ be a real matrix.
    Exactly one of the following statements is true.
    \begin{enumerate}
        \item 
        There is a vector $\vu \geqq \vec{0}$ such that $\vM \vu < \vec{0}$.

        \item
        There is a vector $\vv \geq \vec{0}$ such that $\vv \vM \geqq \vec{0}$.
    \end{enumerate}
\end{theorem}

We require the following discrete variant of \Cref{thm:linear-alternative}.
The geometric intuition of this version is illustrated in \Cref{fig:linear-alternative}.
\opt{submission}{It is proven in the appendix.}

\begin{corollaryrep}
    \label{cor:linear-alternative-integer}
    Let $\vM$ be a rational matrix.
    Exactly one of the following statements is true.
    \begin{enumerate}
        \item
        There is an integer vector $\vu \geq \vec{0}$ such that $\vM \vu \geqq \vec{0}$.
        
        \item 
        There is an integer vector $\vv \geqq \vec{0}$ such that $\vv \vM < \vec{0}$.
    \end{enumerate}
\end{corollaryrep}

\begin{proof}
    For convenience when we use \Cref{cor:linear-alternative-integer} in proving \Cref{thm:potential-function-iff-entirely-execution-bounded}, we swapped the roles of $\vu$ and $\vv$ in left- vs. right-multiplication with $\vM$;
    the real-valued version of the statement of \Cref{cor:linear-alternative-integer} is equivalent to \Cref{thm:linear-alternative} by taking the transpose of $\vM$.

    To see that we may assume the vectors are integer-valued if $\vM$ is rational-valued,
    recall that a system of linear equalities/inequalities with rational coefficients has a solution if and only if it has a rational solution.
    Since the system is homogeneous (the matrix-vector product is compared to the zero vector $\vec{0}$),
    any multiple of a solution is also a solution.
    By clearing denominators,
    it has a rational solution if and only if it has an integer solution.
\end{proof}

\begin{figure}
    \centering
    \includegraphics[width=0.8\linewidth]{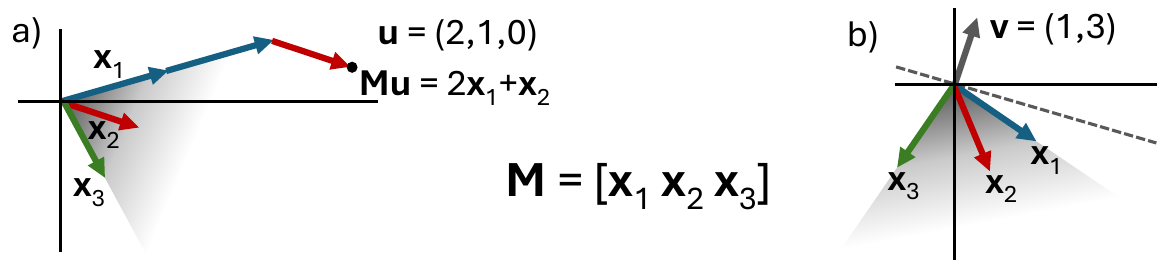}
    \caption{\footnotesize
    Geometric intuition of \Cref{cor:linear-alternative-integer}.
    A matrix $\vM$ has column vectors $\vx_1,\vx_2,\vx_3$.
    The \emph{cone} of $\vM$ is the nonnegative span of these vectors, shown as a faded gray region.
    Exactly one of two scenarios occurs:
    \textbf{a)}
    The cone of $\vM$ intersects the first quadrant (nonnegative orthant in higher dimensions) away from the origin, 
    i.e., some semipositive point ($\vM\vu \geq \vec{0}$ above) is a nonnegative linear combination of vectors $\vx_1,\vx_2,\vx_3$.
    \textbf{b)}
    The cone of $\vM$ does not intersect the first quadrant except at the origin.
    In this case we can draw a dashed line (hyperplane in higher dimensions)
    separating the cone of $\vM$ from the first quadrant.
    The orthogonal vector $\vv \geq \vec{0}$ to this line lies in the first quadrant, but $\vv \vM < 0$ means each vector $\vx_i$ has negative dot product $\vv \cdot \vx_i < 0$, i.e., all $\vx_i$ vectors lie on the \emph{other} side of the line.
    }
    \label{fig:linear-alternative}
\end{figure}

\begin{toappendix}
    Although we do not need the following fact,
    it is worthwhile to observe that,
    if $\vM$ is \emph{integer}-valued (as in our application),
    then the solution $\vu$ or $\vv$ (whichever exists) in \Cref{cor:linear-alternative-integer} has entries that are at most exponential in $\| \vM \|$
    i.e., at most exponential in the sum of absolute values of entries of $\vM$
    (see e.g., \cite{papadimitriou1981complexity}).
    So in particular when we consider $\vM$ having small $O(1)$ size entries,
    this means the solution $\vu$ or $\vv$ has entries that are at most exponential in the number of rows and columns of $\vM$.
    When $\vM$ is a stoichiometric matrix, this corresponds to the number of species and reactions, respectively, of the CRN.
\end{toappendix}

\Cref{cor:linear-alternative-integer} will help us prove the following theorem characterizing CRNs with bounded executions from all configurations.
\Cref{thm:potential-function-iff-entirely-execution-bounded} is used in this paper only to prove 
\Cref{thm:non-collapsing_allvoting_entirely_execution_bounded_cannot_decide_majority_parity,thm:non-collapsing_allvoting_entirely_execution_bounded_decide_almost_constant},
but it may also be of independent interest, since it equates a ``global, infinitary, difficult-to-check'' condition (bounded executions from all configurations) with a ``local, easy-to-check'' condition (having a linear potential function).
}

\opt{final}{The following theorem due to Czerner, Guttenberg, Helfrich, and Esparza, is crucial to proving limitations on execution bounded CRNs such as \Cref{thm:non-collapsing_allvoting_entirely_execution_bounded_cannot_decide_majority_parity,thm:non-collapsing_allvoting_entirely_execution_bounded_decide_almost_constant,}.}

\begin{theorem}[\cite{czerner2024fast}]
    \label{thm:potential-function-iff-entirely-execution-bounded}
    A CRN has a linear potential function if and only if it is entirely execution bounded.
\end{theorem}

\opt{full,submission}{
\begin{proof}
    Let $\calC = (\Lambda,R)$ be a CRN.
    The forward direction is easy:
    assuming $\calC$ has potential function $\Phi$,
    since each reaction decreases $\Phi$ by at least 1
    (see \Cref{rmk:linear-potential-integer-coefficients}),
    from any configuration $\vx$,
    $\calC$ can execute at most $\Phi(\vx)$ reactions while keeping $\Phi$ nonnegative.
    Thus $\calC$ is entirely execution bounded.

    To see the reverse direction, assume that $\calC$ is execution bounded from every configuration,
    and let $\vM$ be the stoichiometric matrix of $\calC$.
    We claim there is no integer vector $\vu \geq \vec{0}$ satisfying $\vM \vu \geqq \vec{0}$;
    for the sake of contradiction suppose otherwise.
    Interpreting $\vu$ as counts of reactions to execute,
    for any sufficiently large configuration $\vx$,
    all reactions in $\vu$ can be applied (in arbitrary order),
    and the vector $\vM \vu$ describes the resulting change in species counts, reaching to configuration $\vy = \vx + \vM \vu$.
    Since $\vM \vu \geqq \vec{0}$,
    this path is self-covering, i.e., $\vy \geqq \vx$.
    But since $\calC$ is execution bounded from every configuration,
    by \Cref{lem:execution-bdd-iff-no-self-covering-path},
    $\calC$ has no self-covering path from any configuration,
    a contradiction.
    This establishes the claim that $\vM \vu \geqq \vec{0}$ has no integer solution $\vu \geq \vec{0}$.
    
    By \Cref{cor:linear-alternative-integer},
    there \emph{is} an integer vector $\vv \geqq \vec{0}$ such that $\vv \vM < \vec{0}$.
    Let $\vv \in \N^\Lambda$ be the coefficients of a linear function $\Phi:\N^\Lambda \to \N$, i.e., $\Phi(\vx) = \vv \cdot \vx$.
    Then the vector $\vv \vM \in \Z^R$ represents the amount $\Phi$ changes by one unit of each reaction, i.e., $\vv \vM(\alpha)$ is the amount $\Phi$ increases when executing reaction $\alpha$ once.
    Since $\vv \vM < \vec{0}$, this means that every reaction strictly decreases $\Phi$, i.e.,
    $\Phi$ is a linear potential function for $\calC$.
\end{proof}
}

\begin{toappendix}
\begin{remark}
    By employing the real-valued version of \Cref{cor:linear-alternative-integer},
    the above proof also shows that \Cref{thm:potential-function-iff-entirely-execution-bounded} holds for the \emph{continuous} model of CRNs~\cite{chen2023rate}, in which species amounts are modeled as continuous nonnegative real concentrations.
    In this case, a continuous CRN would be defined to be execution bounded from configuration $\vx$ if each reaction can be executed by at most a finite (real-valued) amount from $\vx$.
\end{remark}
\end{toappendix}

\subsection{Impossibility of stably deciding majority and parity}

In this section, we prove 
\Cref{thm:non-collapsing_allvoting_entirely_execution_bounded_cannot_decide_majority_parity},
which is a special case of our main negative result, \Cref{thm:non-collapsing_allvoting_entirely_execution_bounded_decide_almost_constant}.
We give a self-contained proof of 
\Cref{thm:non-collapsing_allvoting_entirely_execution_bounded_cannot_decide_majority_parity}
because it is simpler and serves as an intuitive warmup to some of the key ideas used in proving \Cref{thm:non-collapsing_allvoting_entirely_execution_bounded_decide_almost_constant},
without the complexities of dealing with arbitrary semilinear sets.

\Cref{thm:non-collapsing_allvoting_entirely_execution_bounded_cannot_decide_majority_parity} shows a limitation on the computational power of entirely execution bounded, all-voting CRNs, 
but it requires an additional constraint on the CRN for the result to hold
(and we later give counterexamples showing that this extra hypothesis is provably necessary), described in the following definition.

\begin{definition}
\label{defn:non-collapsing}
    Let $\calD$ be a CRD.
    The \emph{output size} of $\calD$ is the function $s:\N \to \N$ defined $s(n) = \min_{\vx,\vy} \{ \|\vy\| \mid \vx \reach \vy, \|\vx\|=n, \vx \text{ is a valid initial configuration}, \vy \text{ is stable} \}$,
    the size of the smallest stable configuration reachable from any valid initial configuration of size $n$.
    A CRD is \emph{non-collapsing} if  $\lim_{n\to\infty}s(n)=\infty$.
\end{definition}

Put another way, $\calD$ is \emph{collapsing} if there is a constant $c$ such that, from infinitely many initial configurations $\vx$,
$\calD$ can reach a stable configuration of size at most $c$.
All population protocols are non-collapsing, since every reaction preserves the configuration size.

\begin{theorem}
\label{thm:non-collapsing_allvoting_entirely_execution_bounded_cannot_decide_majority_parity}
    No non-collapsing, all-voting, entirely execution bounded CRD can stably decide the majority predicate $[X_1 \geq X_2?]$
    or the parity predicate $[X \equiv 1 \mod 2?]$.
\end{theorem}

\begin{proof}
    Let $\calD=(\Lambda,R,\Sigma,\Upsilon_\mathrm{Y}, \Upsilon_\mathrm{N},\vs)$ be a CRD obeying the stated conditions,
    and suppose for the sake of contradiction that $\calD$ stably decides the majority predicate (so $\Sigma=\{X_1,X_2\}$).
    
    We consider the sequence of stable configurations $\va_1,\vb_1,\va_2,\vb_2,\dots$ defined as follows.
    Let $\va_1$ be a stable configuration reachable from initial configuration $\vs+\{X_1, X_2\}$;
    since the correct answer is yes,
    all species present in $\va_1$ vote yes.
    Now add a single copy of $X_2$.
    By additivity,
    the configuration $\va_1 + \{X_2\}$ is reachable from $\vs + \{X_1,2X_2\}$,
    for which the correct answer in this case is no.
    Thus, since $\calD$ stably decides majority,
    from $\va_1 + \{X_2\}$,
    a stable ``no'' configuration is reachable; 
    call this $\vb_1$.
    Now add a single $X_1$.
    Since the correct answer is yes, from $\vb_1 + \{X_1\}$ a stable ``yes'' configuration is reachable,
    call it $\va_2$.

    Continuing in this way, we have a sequence of stable configurations
    $
        \va_1,
        \vb_1, 
        \va_2,
        \vb_2, 
        \dots
    $
    where all species in $\va_i$ vote yes and all species in $\vb_i$ vote no.
    Since $\calD$ is non-collapsing,
    the size of the configurations $\va_i$ and $\vb_i$ increases without bound as $i \to\infty$.
    (Possibly $\|\va_{i+1}\| < \|\va_i\|$, i.e., the size is not necessarily monotonically increasing, but for all sufficiently large $j > i$, we have $\|\va_j\| > \|\va_i\|$.)
    Since all species vote,
    for some constant $\delta > 0$,
    to get from $\va_i+\{X_2\}$ to $\vb_i$,
    at least $\delta \| \va_i \|$ reactions must occur.
    This is because all species in $\va_i$ must be removed since they vote yes,
    and each reaction removes at most $O(1)$ molecules.
    (Concretely, let $\delta = 1 / \max_{(\vr,\vp) \in R} \| \vr \| - \| \vp \|$, i.e., 1 over the most net molecules consumed in any reaction.)
    Similarly,
    to get from $\vb_i+\{X_1\}$ to $\va_{i+1}$, 
    at least $\delta \| \vb_i \|$ reactions must occur.

    Since $\calD$ is entirely execution bounded, by \Cref{thm:potential-function-iff-entirely-execution-bounded},
    $\calD$ has a linear potential function $\Phi(\vx) = \vv \cdot \vx$,
    where $\vv \geq \vec{0}$.
    Adding a single $X_2$ to $\va_i$ increases $\Phi$ by the constant $\vv(X_2)$.
    Since $\|\va_i\|$ grows without bound,
    the number of reactions to get from $\va_i + \{X_2\}$ to $\vb_i$ increases without bound as $i\to\infty$,
    and since each reaction strictly decreases $\Phi$ by at least 1,
    the total change in $\Phi$ that results from adding $X_2$ and then going from $\va_i + \{X_2\}$ to $\vb_i$ is unbounded in $i$, so unboundedly negative for sufficiently large $i$
    (negative once $i$ is large enough that $\delta \| \va_i \| \geq \vv(X_2) + 2$).
    Similarly, adding a single $X_1$ to $\vb_i$ and going from $\vb_i + \{X_1\}$ to $\va_{i+1}$,
    the resulting total change in $\Phi$ is unbounded and (for large enough $i$) negative.

    $\Phi$ starts this process at the constant $\Phi(\vs+\{X_1,X_2\})$.
    Before $\|\va_i\|$ and $\|\vb_i\|$ are large enough that 
    $\delta \| \va_i \| \geq \vv(X_2) + 2$ and 
    $\delta \| \vb_i \| \geq \vv(X_1) + 2$
    (i.e., large enough that the net change in $\Phi$ is negative resulting from adding a single input and going to the next stable configuration),
    $\Phi$ could increase,
    if $\Phi(\{X_1\})$ 
    (resp. $\Phi(\{X_2\})$) 
    is larger than the net decrease in $\Phi$ due to following reactions to get from $\va_i + \{X_2\}$ to $\vb_i$ 
    (resp. from $\vb_i + \{X_1\}$ to $\va_i$).
    
    However, since $\calD$ is non-collapsing,
    this can only happen for a constant number of $i$
    (so $\Phi$ never reaches more than a constant above its initial value $\Phi(\vs+\{X_1, X_2\})$),
    after which $\Phi$ strictly decreases after each round of this process.
    At some point in this process,
    $\calD$ will not be able to reach all the way to the next $\va_i$ or $\vb_i$ without $\Phi$ becoming negative,
    a contradiction.

    The argument for parity is similar, but instead of alternating adding $X_1$ then $X_2$, in each round we always add one more $X$ to flip the correct answer.
\end{proof}

\Cref{thm:non-collapsing_allvoting_entirely_execution_bounded_cannot_decide_majority_parity}
is false without the non-collapsing hypothesis.
The following collapsing, leaderless (but all-voting and entirely execution bounded) CRD stably decides majority:
Species $X_1,x_1$ vote yes, while $X_2,x_2$ vote no:
\begin{align*}
    X_1+X_2 &\rxn x_1+x_2
    \\
    X_1+x_2 &\rxn X_1
    \\
    X_2+x_1 &\rxn X_2
    \\
    x_1+x_2 &\rxn x_1
\end{align*}
It has bounded executions from every configuration:
$\min(\#X_1,\#X_2)$ of the first reaction can occur,
and the other reactions decrease molecular count,
so are limited by the total configuration size.
However, it is collapsing since,
for any $n$, there exists an input of size $n$ that reaches a stable configuration of size 1.
\Cref{thm:non-collapsing_allvoting_entirely_execution_bounded_cannot_decide_majority_parity}
is similarly false without the all-voting hypothesis; for each of the reactions with one product above, add another non-voting product $W$.
This converts the CRD to be non-collapsing but not all-voting.
Of course,
the execution bounded hypothesis is also necessary:
the original population protocols paper~\cite{angluin2004computation} showed that all-voting, non-collapsing, leaderless population protocols can stably decide all semilinear predicates.

The following collapsing, all-voting, leaderless (but entirely execution bounded) CRD stably decides parity.
Let the input species be named $X_1$.
Species $X_1$ votes yes, $X_0$ votes no:
\begin{align*}
    X_1+X_1 &\rxn X_0
    \\
    X_1+X_0 &\rxn X_1
    \\
    X_0+X_0 &\rxn X_0
\end{align*}
\opt{full}{
It has bounded executions from any configuration:
exactly $\# X_1 + \# X_0 - 1$ reactions can occur since each reduces $\#X_1 + \#X_0$ by 1.
Similar to above, by adding the non-voting product $W$ to each reaction above,
this CRD becomes non-collapsing but not all-voting, showing that the all-voting hypothesis is also necessary for stably deciding parity.
}

\subsection{Impossibility of stably deciding not eventually constant predicates}

We now present our main negative result, \Cref{thm:non-collapsing_allvoting_entirely_execution_bounded_decide_almost_constant},
which generalizes \Cref{thm:non-collapsing_allvoting_entirely_execution_bounded_cannot_decide_majority_parity} to show that such CRNs can stably decide only very limited (eventually constant) predicates.

\begin{definition}
\label{defn:eventually_constant_predicate}
    Let $\phi: \N^d \rightarrow \{0, 1\}$ be a predicate. We say $\phi$ is \emph{eventually constant} if there is $n_0 \in \N$ such that $\phi$ is constant on $\N_{\geq n_0}^d=\left\{\vx \in \N^d \mid(\forall i \in\{1, \ldots, d\})\ \vx(i) \geq n_0\right\}$, i.e., either $\phi^{-1}(0) \cap \N_{\geq n_0}^d=\emptyset$ or $\phi^{-1}(1) \cap \N_{\geq n_0}^d=\emptyset$. 
\end{definition}
In other words, although $\phi$ may have an infinite number of each output, ``sufficiently far from the boundary of the positive orthant'' (where all coordinates exceed $n_0$), only one output appears. 
\opt{submission,full}{See \Cref{fig:almost_vs_eventually_vs_not_constant} for a 2D example.}
\opt{submission}{A complete proof appears in the appendix.}
\opt{final}{A complete proof appears in the full version of this paper.}

\begin{toappendix}

\begin{figure}
    \centering
    \includegraphics[width=0.8\linewidth]{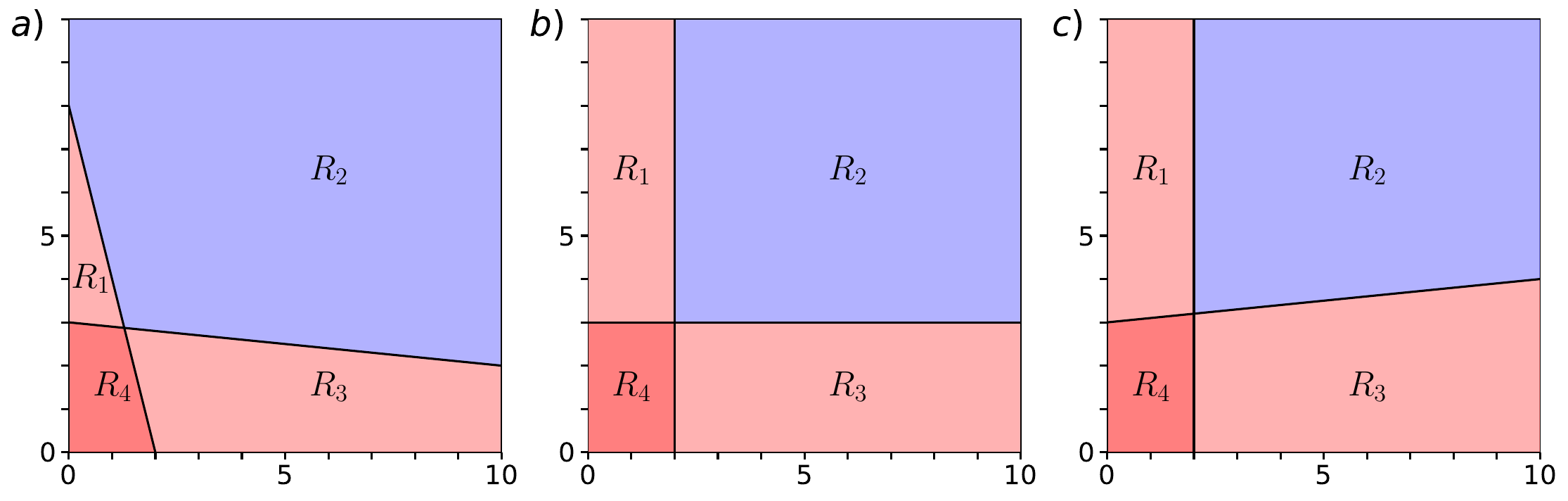}
    \caption{Boolean combinations of threshold predicates in 2D. 
    In all cases, inputs of region $R_2$ output \emph{yes}, while $R_1, R_3, R_4$ output \emph{no}. 
    In part (a), the \emph{no} regions are finite, making the predicate trivially \emph{eventually constant}. 
    Part (b) features both infinite \emph{no} and \emph{yes} regions; setting $m_0 \geq 3$, the output within $\N^k_{\geq m_0}$ is constant, rendering the predicate \emph{eventually constant}.
    In contrast, part (c) has two totally unbounded regions $R_2$ and $R_3$, both of which contain points that are arbitrarily large on all components. 
    We further observe that this implies a separating line with a positive, non-infinite slope between the two regions (in 2D). 
    Generalized to more dimensions: a hyperplane with a normal vector with at least one positive and at least one negative component (\Cref{lem:threshold-not-eventually-constant-regions-two-adjacent-infinite-opposite-output}).}
    \label{fig:almost_vs_eventually_vs_not_constant}
\end{figure}

For any set $B \subseteq \N^d$ and $\vv \in \N^d$, write $B+\vv$ to denote the set $\{ \vx + \vv \mid \vx \in B\},$
which is $B$ translated by vector $\vv$.
Let $\vu_i \in \N^d$ denote the unit vector in direction $i$, i.e., $\vu_i(i)=1$ and $\vu_i(j)=0$ for $j \neq i$.

\begin{definition}
\label{defn:periodic}
    We say $A \subseteq \N^d$ is \emph{periodic} if, for some $k \in \N^+$, for some finite set $F \subseteq \{0,1,\dots,k-1\}^d$,
    $A = \bigcup_{n_1,\dots,n_d \in \N} F + \sum_{i=1}^d k \cdot n_i \cdot \vu_i$.
    We say $k$ is the \emph{period} of $A$ and say that $A$ is $k$-periodic.
    Equivalently, $A$ is \emph{$k$-periodic} if, for all $\vx \in \N^d$ and all unit vectors $\vu_i$, $\vx \in A \iff \vx + k \cdot \vu_i \in A$.
\end{definition}

In other words, $A$ is periodic if it is a union of copies of a finite subset $F$ of the $k \times k \times \dots \times k$ hypercube with a corner at the origin,
translated in each direction by every nonnegative integer multiple of the hypercube's width $k$.
See \Cref{fig:periodic_set_2d}.
Note that if $A$ is $k$-periodic, then it is also $k'$-periodic for every positive integer multiple $k' = i \cdot k$ of $k$.

\begin{figure}
    \centering
    \includegraphics[width=0.8\linewidth]{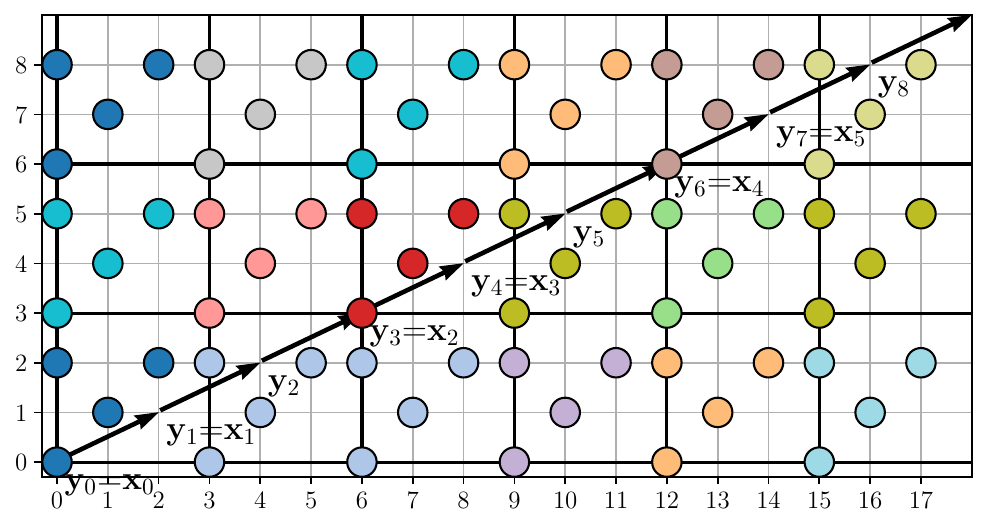}
    \caption{A 2D periodic set $A$.
    The finite set $F = \{(0,0), (0,2), (1,1), (2,2)\}$, a subset of the $3 \times 3$ square $\{0,1,2\}^2$, is translated in both $x$ and $y$ directions by every nonnegative vector with entries that are multiples of 3.
    If the set is periodic but not constant, 
    then there is some vector $\vv$ such that the sequence $\vy_0,\vy_1,\dots$,
    each being the previous plus $\vv$
    ($\vv=(2,1)$ in this example),
    is periodic but not constant,
    and then a regularly-spaced subsequence $\vx_0,\vx_1,\dots$,
    alternating points in $A$ and points not in $A$,
    satisfies \Cref{lem:semilinear-not-almost-constant-alternating-period-vector-and-offset}.
    For brevity,
    the figure shows $\vy_0$ as the origin, but in the proof of \Cref{lem:semilinear-not-almost-constant-alternating-period-vector-and-offset},
    $\vy_0$ is chosen to be sufficiently far from the origin that, moving in direction $\vv$ infinitely from that point will not cross any hyperplanes defined by the threshold sets of that proof.
    }
    \label{fig:periodic_set_2d}
\end{figure}

\begin{lemma}
\label{lem:mod-sets-periodic-hypercube}
    Let $A \subseteq \N^d$ be a Boolean combination of mod sets.
    Then $A$ is periodic.
\end{lemma}

\begin{proof}
    We prove this by induction on the number of mod sets.
    For the base case, let 
    $A = \{ \vx \mid \vw \cdot \vx \equiv c \mod m \}$ be a single mod set,
    where $\vw \in \{0,\dots,m-1\}^d$ and $c,m \in \N$ are constants.
    Letting $k=m$ 
    and
    $F = A \cap \{0,\dots,m-1\}^d$
    in \Cref{defn:periodic} works.
    Let $\vx \in \N^d$.
    Then for all $1 \leq i \leq d$,
    $\vw \cdot \vx \equiv \vw \cdot (\vx + m \vu_i) \mod m$,
    so 
    $\vw \cdot \vx \equiv c \mod m \iff \vw \cdot (\vx + m \vu_i) \equiv c \mod m$,
    meaning that $\vx \in A \iff \vx + m \vu_i \in A$,
    so $A$ is $k$-periodic.

    The inductive case amounts to showing that periodic sets are closed under Boolean operations of union, intersection, and complement.
    Clearly the complement of any periodic set is also periodic.
    
    Inductively assume that $A_1,A_2 \subseteq \N^d$ are periodic;
    we argue that $A_1 \cup A_2$ is periodic.
    Letting $k$ be the least common multiple of their periods, we may assume both $A_1$ and $A_2$ are $k$-periodic with the same period $k$.
    Then for all $\vx \in \N^d$
    and all unit vectors $\vu_i$,
    $\vx \in A_1 \iff \vx + k \cdot \vu_i \in A_1$
    and
    $\vx \in A_2 \iff \vx + k \cdot \vu_i \in A_2$.
    Thus
    $\vx \in A_1 \cup A_2 \iff \vx + k \cdot \vu_i \in A_1 \cup A_2$,
    so $A_1 \cup A_2$ is also $k$-periodic.
    Similar reasoning shows $A_1 \cap A_2$ is $k$-periodic (one can also appeal to DeMorgan's Laws).
\end{proof}

Each threshold set $T$ is defined by a hyperplane that partitions $\N^d$ into the sets $T$ 
(on one side of the hyperplane, including integer points on the hyperplane itself) 
and $\overline{T}$ 
(on the other side of the hyperplane).
More generally, several threshold sets partition $\N^d$ into multiple disjoint subsets we call ``regions''.
Furthermore, any predicate that is a Boolean combination of threshold sets has constant output in any region; the next definition formalizes this.

\begin{definition}
\label{defn:threshold-regions}
    Let $A \subseteq \N^d$ be Boolean combination of threshold sets $T_1,\dots,T_k \subset \N^d$.
    A \emph{region} of $A$ is a convex polytope $R \subset \R_{\geq 0}^d$ such that, for all $\vx,\vy \in R \cap \N^d$,
    for all $1 \leq i \leq k$,
    $\vx \in T_i \iff \vy \in T_i$.
    The \emph{output} of the region $R$ is the value 1 if $R \cap \N^d \subset A$ and 0 if $R \cap \N^d \cap A = \emptyset$.
    (Note these are the only two possibilities, since no individual threshold set $T_i$ is exited or entered as we move within $R$.)
    A region $R$ is \emph{totally unbounded}
    if, for all $c \in \N$,
    $R \cap \N_{\geq c}^d \neq \emptyset$,
    i.e., $R$ contains points that are arbitrarily large on all components.
    A region is called \emph{partially bounded} if it is not totally unbounded.
\end{definition}

Put another way, predicates defined by Boolean combinations of threshold sets are defined by $(d-1)$-dimensional hyperplanes that partition $\N^d$ into regions, where in each region, the output of the predicate is all yes, or all no.
In fact this is an exact characterization of Boolean combinations of threshold predicates.

\begin{definition}
    For any set $A \subseteq \R^d$,
    the \emph{recession cone} of $A$ is
    \[
        \recc(A) = \{ \vv \in \R^d \mid (\forall \vx \in A)(\forall \lambda > 0)\ \vx + \lambda \vv \in A \},
    \]
    the set of vectors $\vv$ such that, 
    from any point in $A$, one can move in direction $\vv$ forever without leaving $A$.
\end{definition}

\begin{observation}
\label{obs:region-totally-unbounded-iff-recc-contains-positive-vector}
    A region $R$ defined by threshold sets is totally unbounded if and only if $\recc(R) \cap \R_{> 0}^d \neq \emptyset$,
    i.e., the recession cone of $R$ contains a positive vector.    
\end{observation}


\begin{lemma}
\label{lem:threshold-not-eventually-constant-regions-two-adjacent-infinite-opposite-output}
    Let $A \subseteq \N^d$ be Boolean combination of threshold sets that is not eventually constant.
    Then there are two adjacent totally unbounded regions $R_0$, $R_1$ with opposite outputs, such that the normal vector $\vh$ of the hyperplane $H$ separating $R_0$ and $R_1$ has at least one negative component and at least one positive component.
\end{lemma}

\begin{proof}
    See \Cref{fig:threshold_sets} for an example in 2D.
    Since $A$ is not eventually constant, it must have two totally unbounded regions $R_0$ and $R_1$ with opposite outputs;
    assume WLOG that $R_i$ has output $i$.
    Let $c \geq 0$ be sufficiently large that all partially bounded regions of $A$ are subsets of $B = \N^d \setminus \N_{\geq c}^d$.
    Now, simply pick any points $\vx_0 \in R_0 \setminus B$ and $\vx_1 \in R_1 \setminus B$.
    There is some path from $\vx_0$ to $\vx_1$ that follows only unit vectors (i.e., moves only to adjacent points that are distance 1 from the previous point),
    such that every intermediate point $\vx'$ also obeys $\vx' \not \in B$.
    
    Then this path never enters a partially bounded region of $A$,
    since they are all subsets of $B$.
    Thus, since the path starts in a region $R_0$ with output 0,
    ends in a region $R_1$ with output 1,
    there must be two adjacent points $\va,\vb$ on the path,
    where $\va$ is in a totally unbounded region with output 0 and $\vb$ is in a totally unbounded region with output 1.

    Finally, we must that the normal vector of the hyperplane separating $R_0$ from $R_1$ has a negative and a positive entry. 
    Recall that a threshold set $T$ is defined by $T = \{ \vx \in \mathbb{N}^d \mid \vw \cdot \vx \leq a \}$, where $\vw = (w_1, \dots, w_d) \in \N^d$ and $a \in \N$ (\cref{defn:threshold_and_mod_set}). 
    Since $A$ is a Boolean combination of threshold sets and $R_0, R_1$ are adjacent with opposite outputs, there must be some threshold set $T$ such that $R_1 \subseteq T$, but $R_0 \cap T = \emptyset$ (or vice versa, but assume $R_1 \subseteq T$ WLOG, since we could replace $T$ with $\overline{T}$ in the Boolean combination defining $A$).
    Equivalently, we can think of the regions $R_0$ and $R_1$ as being separated by the hyperplane $\vw \cdot \vx = a$, with normal vector $\vw$ and offset $a$, such that all points $\vx \in R_1$ obey $\vw \cdot \vx \leq a$, 
    and all points $\vx \in R_0$ obey $\vw \cdot \vx > a$.
    The transition between the regions at points $\va$ and $\vb$ involves crossing the hyperplane, where the inequality changes from $\leq a$ to $> a$, which defines the boundary between different outputs (0 in $R_0$ and 1 in $R_1$). Therefore, the points on the hyperplane $\vw \cdot \vx = a$ necessarily lie exactly at the boundary between these regions.
    
    We show that $\vw$ cannot be nonnegative or nonpositive. Suppose $\vw \geq \vec{0}$ (scale the normal vector by $-1$ otherwise). 
    Since $R_1$ is totally unbounded, it contains points that are arbitrarily large on all components. 
    More formally, there is a strictly increasing sequence $\vx_1 < \vx_2 < \dots$ such that all $\vx_i \in R_1.$
    Since $\vw \geq \vec{0}$, 
    $\lim_{i \to \infty} \vw \cdot \vx_i = \infty$.
    This contradicts the previous assumption that
    all points $\vx \in R_1$
    obey
    $\vw \cdot \vx \leq a$ (geometrically, we would cross the hyperplane somewhere and land in $R_0$).
    Symmetric reasoning applies to the case $\vw \leq \vec{0}$.
    We conclude that the separating hyperplane must have a normal vector $\vw$ with at least one positive and at least one negative component, establishing the lemma.
\end{proof}

The next lemma shows that the there exists a vector $\vv > \vec{0}$ parallel to the hyperplane separating the two regions. In other words, we can move along $H$ while increasing every component.

\begin{lemmarep}
\label{lem:normal_vector_hyperplane_pos_neg_for_pos_vec_parallel}
    Let $H$ be a hyperplane with normal vector $\vh$.
    Then there is a positive vector $\vv > \vec{0}$ with $\vv \cdot \vh = 0$
    if and only if $\vh$ has at least one negative component and at least one positive component.
\end{lemmarep}

\begin{proof}
\item[$\implies$:]
    If $\vv > \vec{0}$ and $\vh \geq \vec{0}$ then $\vv \cdot \vh > 0$.
    Similarly, if $\vv > \vec{0}$ and $\vh \leq \vec{0}$ then $\vv \cdot \vh < 0$.
    So to get $\vv \cdot \vh = 0$, $\vh$ must have at least one positive and at least one negative element.

\item[$\impliedby$:]
    We construct $\vv$ as follows: Let $I_+$ denote the indices of the positive coordinates of $\vh$ and $I_-$ the indices of the negative coordinates. Our goal is to balance out the positive and negative parts of the dot product, given by $\vv \cdot \vh=\sum_{i \in I_{+}} \vv(i)\vh(i)+\sum_{i \in I_{-}} \vv(i)\vh(i)$. Set $\vv(i)$ to be the sum of the positive coordinates of $\vh$ if $i \in I_{-}$ and the sum of the absolute values of negative coordinates of $\vh$ otherwise:
    \begin{align*}
        \vv(i) =
        \begin{cases} 
        \sum_{j \in I_-} |\vh(j)| & \text{if } i \in I_+, \\
        \sum_{j \in I_+} \vh(j) & \text{if } i \in I_-, \\
        0 & \text{otherwise.}
        \end{cases}
    \end{align*}
    Substituting into the formula shows the correctness. For brevity, let $p := \vv(i)$ if $i \in I_{+}$ and $n := \vv(i)$ if $i \in I_{-}$ as above.
    \begin{align*}
        \vv \cdot \vh 
        &= 
        \sum_{i \in I_{+}} \vv(i)\vh(i) + \sum_{i \in I_{-}} \vv(i)\vh(i) 
        \\ &= 
        \sum_{i \in I_{+}} \left( n \right) \vh(i) + \sum_{i \in I_{-}} \left( p \right) \vh(i) 
        \\ &= 
        n \sum_{i \in I_{+}} \vh(i) + p \sum_{i \in I_{-}} \vh(i)
        \\ &= 
        n p + p (-n)
        \\ &= 
        0.
\end{align*}
    Finally, if $\vv$ is not integer-valued, scale it by the least common multiple of all coordinate denominators to ensure $\vv \in \N^d$ without altering the dot product.
    \qedhere

\end{proof}


\begin{lemma}
\label{lem:semilinear-not-almost-constant-alternating-period-vector-and-offset}
    Let $\phi:\N^d \to \{0,1\}$ be a semilinear predicate that is not eventually constant.
    Then there is an infinite sequence $\vx_0, \vx_1, \dots$ and constant $c$, 
    such that for all $j \in \N$,
    \begin{enumerate}
        \item
        \label{lem:semilinear-not-almost-constant-alternating-period-vector-and-offset:cond1}
        $\phi(\vx_j) \neq \phi(\vx_{j+1})$
        (correct answer swaps for each subsequent input),

        \item
        \label{lem:semilinear-not-almost-constant-alternating-period-vector-and-offset:cond2}
        $\vx_j \leq \vx_{j+1}$ 
        (inputs are increasing),
        and

        \item
        \label{lem:semilinear-not-almost-constant-alternating-period-vector-and-offset:cond3}
        $\|\vx_{j+1} - \vx_j\| \leq c$ 
        (adjacent inputs are ``close'').
    \end{enumerate}
\end{lemma}

\begin{proof}
    We associate to $\phi$ the set $A \subseteq \N^d$ where $\phi^{-1}(1) = A$, i.e., $\phi(\vx) = 1 \iff \vx \in A$.
    
    Since $A$ is semilinear, it is a Boolean combination of threshold sets $T_1,\dots,T_k$ and mod sets $M_1,\dots,M_l$.
    Recall \Cref{defn:threshold-regions}, where the threshold sets partition $\N^d$ into \emph{regions},
    where moving within a region does not cross and hyperplanes defining the threshold sets,
    thus does not change the Boolean value [$\vx \in T_i$?] for any $T_i$.
    Suppose we have $m$ regions $R_1,\dots,R_m.$
    Then we can rewrite $A \cap R_j$ as a Boolean combination of mod sets only, intersected with $R_j$.
    We do this by replacing each $T_i$ in the original Boolean expression 
    with either $\N^d$ or $\emptyset$, depending whether $R_j \subseteq T_i$ or $R_j \cap T_i = \emptyset$, respectively.\footnote{
        For example, if the expression is 
        $T_1 \cup (M_1 \cap T_2) \cup (M_2 \cup T_3) \cup (M_3 \cap M_4)$,
        if the points are in $T_2$ but not $T_1$ or $T_3$,
        this becomes
        $\emptyset \cup (M_1 \cap \N^d) \cup (M_2 \cup \emptyset) \cup (M_3 \cap M_4) = M_1 \cup M_2 \cup (M_3 \cap M_4)$.
    }
    (Note by the definition of region these are the only two possibilities.)
    Let $M'_j$ be this Boolean combination of mod sets, such that $M'_j \cap R_j = A \cap R_j$.
    By \Cref{lem:mod-sets-periodic-hypercube},
    $M'_j$ is periodic.

    Consider a totally unbounded region $R_j.$
    By \Cref{obs:region-totally-unbounded-iff-recc-contains-positive-vector},
    $\recc(R_j)$ contains a positive vector $\vv$.
    We have two cases:
    \begin{description}

    \item[for some totally unbounded region $R_j$, $M'_j \cap R_j$ is not constant:]
        This is illustrated in \Cref{fig:periodic_set_2d,fig:region_mod_periodic_not_constant_but_constant_along_recession_cone},
        which show two subcases.
        \Cref{fig:periodic_set_2d}
        shows the subcase where,
        for some $\vv \in \recc(R_j) \cap \N^d$ and point $\vy_0 \in R_j$,
        defining $\vy_i = \vy_0 + i \vv$,
        the sequence $\phi(\vy_0),\phi(\vy_1),\dots$ is not constant.
        Since $M'_j$ is periodic,
        the sequence $O = (\phi(\vy_0),\phi(\vy_1),\dots)$ is periodic with period $p$.
        So we can find a subsequence $\vx_0,\vx_1,\dots$ obeying all three conditions of the lemma.
        In particular, it suffices to choose a point $\vx_0 = \vy_0 \in R_j \cap A$ 
        (resp. $\vy_0 \in R_j \setminus A$ )
        let $i < p$ such that
        $\vy_i \not \in A$
        (resp. $\vy_i \in A$),
        letting $\vx_1 = \vy_i$,
        and let $\vx_2 = \vy_p$,
        and subsequent elements of the subsequence are the same distances apart 
        ($\vx_3 = \vy_{p+i},
        \vx_4 = \vy_{2p}, \dots$).

        \Cref{fig:region_mod_periodic_not_constant_but_constant_along_recession_cone} 
        shows the subcase where,
        for all $\vy_0 \in R_j$ and $\vv \in \recc(R_j) \cap \N^d$,
        defining $\vy_i = \vy_0 + i \vv$,
        the sequence $\phi(\vy_0),\phi(\vy_1),\dots$ is constant.
        However,
        since $M'_j$ is not constant,
        we can still find a sequence $\vx_0,\vx_1,\dots$, but unlike the previous subcase, it is not a subsequence of points collinear along one vector $\vv$.
        
        Since $M'_j$ is periodic and not constant,
        and since $R_j$ is totally unbounded,
        for every $\vx \in R_j \cap A$,
        there is $\vx' \geq \vx$ such that $\vx' \in R_j \setminus A$,
        i.e., for every point in the region in $A$, there is a larger point in $R_j$ not in $A$.
        Also, since $M'_j$ is periodic,
        there is a constant $c$ independent of $\vx$ such that $\| \vx' - \vx \| \leq c$.
        By symmetric reasoning, there is a $\vx'' \in R_j \cap A$ such that $\vx'' \geq \vx'$ and $\| \vx'' - \vx' \| \leq c$.

        Let $\vx_0 \in R_j \cap A$ be arbitrary.
        For all $i \in \N$,
        choose $\vx_{i+1} \in R_j$ based on $\vx_i$ as above,
        such that 
        $\vx_{i+1} \geq \vx_i$, 
        $\| \vx_{i+1} - \vx_i \| \leq c$,
        and
        $\vx_{i+1} \in A$ if $i$ is odd and $\vx_{i+1} \not \in A$ if $i$ is even.
        Then the sequence $\vx_0,\vx_1,\dots$ satisfies the lemma.

    \item[for all totally unbounded regions $R_j$, $M'_j \cap R_j$ is constant:]
        This implies that the mod sets $M_1,\dots,M_l$ can be ``factored out'' of the Boolean expression defining $A$ in terms of threshold sets $T_1,\dots,T_k$ and the mod sets $M_1,\dots,M_l$,
        which will give the same output as $A$ in totally unbounded regions.
        Put another way,
        $A \cap (R_1 \cup \dots \cup R_u)$ is a Boolean combination of the threshold sets $T_1,\dots,T_k$,
        where $R_1,\dots,R_u$ represents all the totally unbounded regions.

        By \Cref{lem:threshold-not-eventually-constant-regions-two-adjacent-infinite-opposite-output},
        two adjacent totally unbounded regions of $A$ have opposite outputs.
        See \Cref{fig:threshold_sets} for an example of picking the points $\vx_0,\vx_1,\dots$ below.
        These adjacent regions are separated by some hyperplane $H_j$, such that $H_j \subseteq A$,
        but for some unit vector $\vu_i$,
        $(H_j + \vu_i) \cap A = \emptyset$,
        i.e., all of $H_j$ is contained in $A$,
        but the entire hyperplane adjacent to $H_j$ in direction $\vu_i$,
        consists of points not in $A$.
        Note this is not true for general hyperplanes, e.g., one whose orthogonal vector is $(1,1)$, where both unit vectors $\vu_1 = (1,0)$ and $\vu_2 = (0,1)$ would move off the hyperplane, but in the ``yes'' direction where the point is still contained in the threshold set.
        However, since $H_j$ is separating two totally unbounded regions,
        some strictly positive vector $\vv > \vec{0}$ is parallel to $H_j$, i.e., obeys $\vv \cdot \vh = 0$ for $H_j$'s orthogonal vector $\vh$.
        By \Cref{lem:normal_vector_hyperplane_pos_neg_for_pos_vec_parallel}, 
        $\vh$ has at least one positive coordinate (say $i$) and at least one negative coordinate (say $k$),
        so that unit vector $\vu_i$ moves to one side of $H_j$ and $\vu_k$ moves to the other side.
        
        In this case,
        we let $\vv > \vec{0}$ be some vector parallel to $H_j$,
        let $\vx_0 \in H_j$,
        sufficiently large that the vector $\vv$, starting at $\vx_0$,
        does not cross any of the hyperplanes of $T_1,\dots,T_k$
        (as in \Cref{fig:threshold_sets}).
        Define the rest of the infinite sequence as
        \begin{align*}
        \vx_1 &= \vx_0 + \vu,
        \\
        \vx_2 &= \vx_0 + \vv,
        \\
        \vx_3 &= \vx_0 + \vv + \vu,
        \\
        \vx_4 &= \vx_0 + 2\vv,
        \\
        \vx_5 &= \vx_0 + 2\vv + \vu,
        \\
        \vx_6 &= \vx_0 + 3\vv,
        \\
        \vx_7 &= \vx_0 + 3\vv + \vu,
        \\
        \vdots
        \end{align*}
        By the arguments given above,
        for all odd $i$,
        $\phi(\vx_i) = 0$
        and for all even $i$,
        $\phi(\vx_i) = 1$, 
        satisfying condition \eqref{lem:semilinear-not-almost-constant-alternating-period-vector-and-offset:cond1}.
        If $j$ is even, then $\vx_{j+1} = \vx_j + \vv + \vu$, so clearly $\vx_j \leq \vx_{j+1}$,
        satisfying condition \eqref{lem:semilinear-not-almost-constant-alternating-period-vector-and-offset:cond2}.
        If $j$ is odd,
        then $\vx_{j+1} = \vx_j - \vu + \vv$.
        Since $\vv > \vec{0}$,
        we have $\vv - \vu \geq \vec{0}$,
        so $\vx_{j+1} \geq \vx_j$,
        satisfying condition \eqref{lem:semilinear-not-almost-constant-alternating-period-vector-and-offset:cond2}.
        Finally,
        $\| \vx_{j+1} - \vx_j \| \leq \| \vv \| + 1$, 
        satisfying condition \eqref{lem:semilinear-not-almost-constant-alternating-period-vector-and-offset:cond3}.
        \qedhere
    \end{description}
\end{proof}

\begin{figure}
    \centering
    \includegraphics[width=0.7\linewidth]{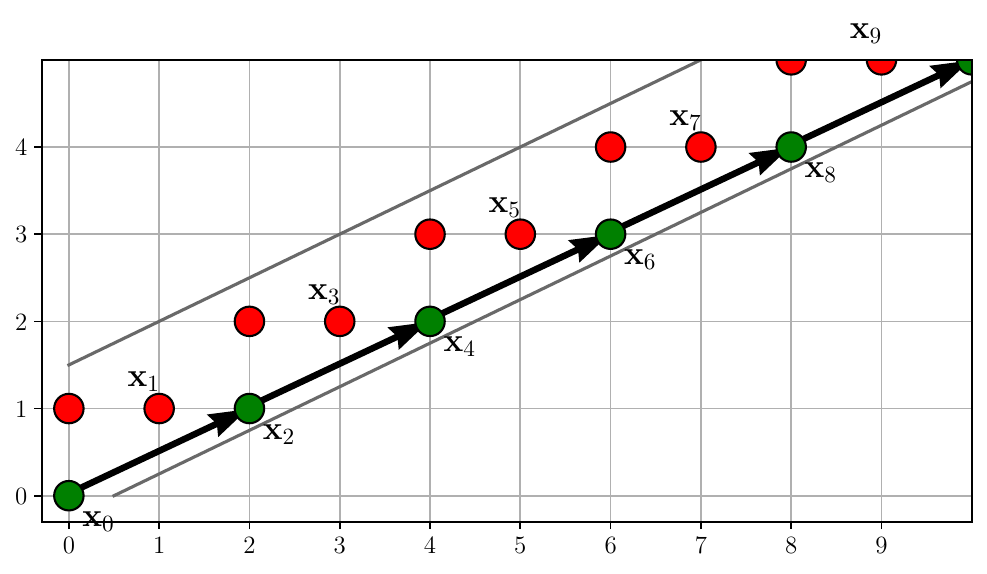}
    \caption{A totally unbounded region $R_j$ (between two gray lines with slope $1/2$ defined by two threshold sets),
    where within $R_j$, the Boolean combination $M'_j$ of mod sets that is not constant (see proof of \Cref{lem:semilinear-not-almost-constant-alternating-period-vector-and-offset} for definitions of $R_j$ and $M'_j$), but that is constant along any vector in $\recc(R)$,
    which contains only multiples of the vector $(2,1)$.
    This illustrates one case of \Cref{lem:semilinear-not-almost-constant-alternating-period-vector-and-offset},
    where we can find an increasing sequence of points $\vx_0,\vx_1,\dots,$ with alternating outputs.
    }    \label{fig:region_mod_periodic_not_constant_but_constant_along_recession_cone}
\end{figure}

\begin{figure}
    \centering
    \includegraphics[width=0.4\linewidth]{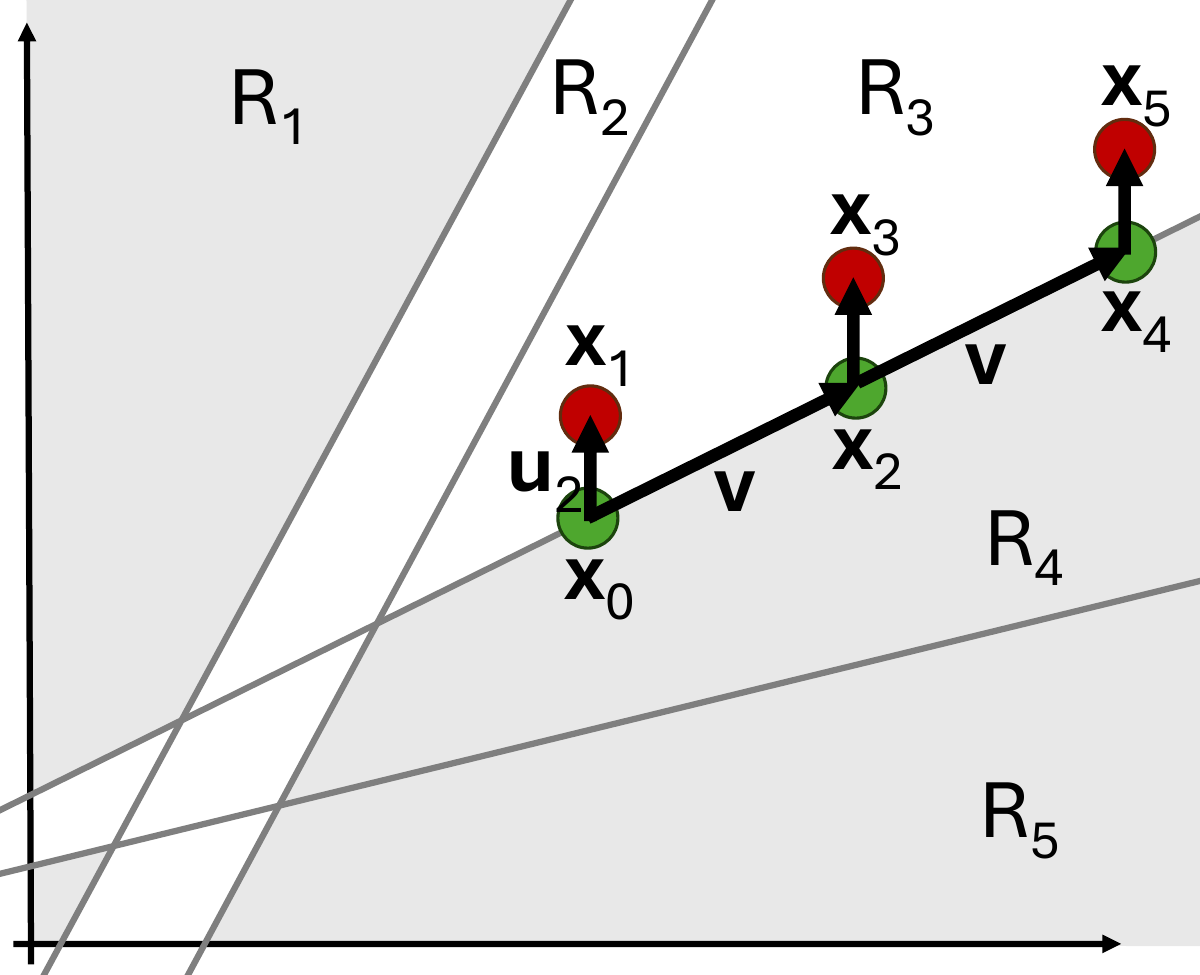}
    \caption{A Boolean combination of threshold sets in 2D, illustrating the intuition of \Cref{lem:threshold-not-eventually-constant-regions-two-adjacent-infinite-opposite-output}.
    Each threshold set is defined by a line (hyperplane in higher dimensions).
    The threshold sets partition the first quadrant into a finite number of \emph{regions}.
    There are four (unlabeled) finite regions near the origin.
    There are five infinite regions labeled $R_1,R_2,R_3,R_4,R_5.$
    Within each region, either all points are contained in the set (output 1, shaded gray regions $R_1,R_4,R_5$), or all points are not (output 0, shaded white regions $R_2,R_3$).    
    Since the set is not eventually constant,
    there are infinite regions with opposite outputs, and two such regions must be adjacent; $R_3$ and $R_4$ in this example.
    The points $\vx_0,\vx_1,\dots$ illustrate the intuition of the last paragraph of the proof of \Cref{lem:semilinear-not-almost-constant-alternating-period-vector-and-offset}.
    The vector $\vv$ is parallel to a line separating the two adjacent opposite-output regions $R_3$ and $R_4$.
    Points $\vx_0,\vx_2,\vx_4,\dots$ are on the line, thus are in $R_4$ and have output 1,
    and points 
    $\vx_1,\vx_3,\dots$
    are just above the line (off of it by unit vector $\vu_2$),
    thus are in $R_3$ and have output 0.
    }
    \label{fig:threshold_sets}
\end{figure}

\end{toappendix}

\begin{thmrep}
\label{thm:non-collapsing_allvoting_entirely_execution_bounded_decide_almost_constant}
    If a non-collapsing, all-voting, entirely execution bounded CRD stably decides a predicate $\phi$,
    then $\phi$ is eventually constant. 
\end{thmrep}

\begin{proofsketch}
    This proof is similar to that of \Cref{thm:non-collapsing_allvoting_entirely_execution_bounded_cannot_decide_majority_parity}.
    In that proof, we repeatedly add a ``constant amount of additional input $\{X_2\}$ or $\{X_1\}$, which flips the output''.
    For more general semilinear, but not eventually constant, predicates,
    we dig into the structure of the semilinear set to find a sequence of constant-size vectors representing additional inputs that flip the correct output.
    Any predicate that is not eventually constant has infinitely many yes inputs and infinitely many no inputs, but in general they could be increasingly far apart:
    e.g., $\phi(\vx) = 1$ if and only if $2^{n} \leq \| \vx \| < 2^{n+1}$ for even $n$.
    For the potential function argument to work, 
    each subsequent input needs to be at most a constant larger than the previous.
    
    But if $\phi$ is \emph{semilinear} (and not eventually constant) then we can show that there is a sequence of increasing inputs $\vx_0 \leq \vx_1 \leq \vx_2 \leq \dots$, each a \emph{constant} distance from the next ($\| \vx_{j+1} - \vx_j \| = O(1)$),
    flipping the output ($\phi(\vx_j) \neq \phi(\vx_{j+1})$).
    Roughly, this is true for one of two reasons.
    Using \Cref{thm:semilinear-Boolean-combination-threshold-mod},
    $\phi$ is a Boolean combination of threshold and mod sets.
    Either the mod sets are not combined to be trivially $\emptyset$ or $\N^d$,
    in which case we can find some vector $\vv$ that, followed infinitely far from some starting point $\vx_0$ (so $\vx_i = \vx_0 + i \vv$)
    periodically hits both yes inputs ($\phi(\vx_j)=1$) and no inputs ($\phi(\vx_j)=0$).
    \opt{submission,full}{(See \Cref{fig:periodic_set_2d,fig:region_mod_periodic_not_constant_but_constant_along_recession_cone}.)}
    Otherwise,
    the mod sets can be removed and simplify the Boolean combination to only threshold sets,
    in which case the infinite sequence $\vx_0,\vx_1,\dots$ can be obtained by moving along a threshold hyperplane that separates yes from no inputs.
    \opt{submission,full}{(See \Cref{fig:threshold_sets}.)}
\end{proofsketch}

\begin{proof}
    This proof is similar to that of \Cref{thm:non-collapsing_allvoting_entirely_execution_bounded_cannot_decide_majority_parity},
    with the vectors $\vv_i$ defined below playing the role of the ``constant amount of additional input $\{X_2\}$ or $\{X_1\}$ that flips the correct answer'' in that proof.

    Let $\calD=(\Lambda,R,\Sigma,\Upsilon_\mathrm{Y}, \Upsilon_\mathrm{N},\vs)$ be a CRD obeying the stated conditions,
    and suppose for the sake of contradiction that $\calD$ stably decides a semilinear predicate $\phi$ that is not eventually constant.
    
    By \Cref{lem:semilinear-not-almost-constant-alternating-period-vector-and-offset},
    there is an infinite sequence $\vx_0, \vx_1, \dots$ such that
    \begin{enumerate}
    \item
        $\phi(\vx_i) \neq \phi(\vx_{i+1})$ (correct answer swaps for each subsequent input),

    \item
        $\vx_i \leq \vx_{i+1}$, (inputs are increasing (on at least one coordinate(s)),
        and

    \item
        for some constant $c$, $\|\vx_{i+1} - \vx_i\| \leq c$. (adjacent inputs are ``close'')
    \end{enumerate}
    Assume WLOG that $\phi(\vx_0) = 0$.
    For each $i \in \N$,
    let $\vv_i = \vx_{i+1} - \vx_i$,
    noting by condition \eqref{lem:semilinear-not-almost-constant-alternating-period-vector-and-offset:cond2} that
    $\vv_i \geq 0$.

    We consider the sequence of stable configurations $\va_0,\va_1,\va_2,\dots$ defined as follows.
    Let $\va_0$ be a stable configuration reachable from $\vx_0$;
    since the correct answer is no,
    all species present in $\va_0$ vote no.
    Now add $\vv_0$ to $\va_0$.
    By additivity,
    the configuration $\va_0 + \vv_0$ is reachable from $\vx_1 = \vx_0 + \vv_0$.
    Since the correct answer for $\vx_1$ is yes,
    $\calD$ must go from $\va_0 + \vv_0$ to a stable ``yes'' configuration, call this $\va_1$.
    Now add $\vv_1$ to $\va_1$.
    Since the correct answer is no, $\calD$ must now reach from $\va_1 + \vv_1$ to a stable ``no'' configuration,
    call it $\va_2$.
    By condition \eqref{lem:semilinear-not-almost-constant-alternating-period-vector-and-offset:cond3},
    each $\vv_i$ obeys $\| \vv_i \| < c$ for some constant $c$.

    Continuing in this way, we have a sequence of stable configurations
    $
        \va_0,
        \va_1, 
        \dots
    $
    where all species in $\va_i$ vote yes for odd $i$, and all species in $\va_i$ vote no for even $i$.
    Since $\calD$ is non-collapsing,
    the size of the configurations $\va_i$ increases without bound as $i \to\infty$.
    (Possibly $\|\va_{i+1}\| < \|\va_i\|$, i.e., the size is not necessarily monotonically nondecreasing, but for all sufficiently large $j > i$, we have $\|\va_j\| > \|\va_i\|$.)
    
    Since all species vote,
    for some constant $\delta > 0$,
    to get from $\va_i+\vv_{i}$ to $\va_{i+1}$,
    at least $\delta \| \va_i \|$ reactions must occur.
    This is because all species in $\va_i$ must be removed since they vote the opposite of the voters in $\va_{i+1}$,
    and each reaction removes at most $O(1)$ molecules.
    (Concretely, let $\delta = 1 / \max_{(\vr,\vp) \in R} \| \vr \| - \| \vp \|$, i.e., 1 over the most net molecules consumed in any reaction.)

    Since $\calD$ is entirely execution bounded, by \Cref{thm:potential-function-iff-entirely-execution-bounded},
    $\calD$ has a linear potential function $\Phi(\vx) = \vw \cdot \vx$,
    where $\vw \geq \vec{0}$.
    Adding $\vv_{i}$ to $\va_i$ increases $\Phi$ by $\vw(\vv_{i})$, which is bounded above by a constant since $\| \vv_{i} \| < c$.
    Since $\|\va_i\|$ grows without bound,
    the number of reactions to get from $\va_i + \vv_{i}$ to $\va_{i+1}$ increases without bound as $i\to\infty$,
    and since each reaction strictly decreases $\Phi$ by at least 1,
    the total change in $\Phi$ that results from adding $\vv_i$ and then going from $\va_i + \vv_{i}$ to $\va_{i+1}$ is unbounded in $i$, so unboundedly negative for sufficiently large $i$
    (negative once $i$ is large enough that $\delta \| \va_i \| \geq \vw(\vv_{i}) + 2$).

    However, $\Phi$ started at the constant $\Phi(\vx_0)$.
    Before $\|\va_i\|$ is large enough that 
    $\delta \| \va_i \| \geq \vw(\vv_{i}) + 2$ 
    (i.e., large enough that the net change in $\Phi$ is negative resulting from adding a single input and going to the next stable configuration),
    $\Phi$ could increase,
    if $\Phi(\vv_{i})$ 
    is larger than the net decrease in $\Phi$ due to following reactions to get from $\va_i + \vv_{i}$ to $\va_{i+1}$.
    
    However, since $\calD$ is non-collapsing,
    this can only happen for a constant number of $i$
    (so $\Phi$ never reaches more than a constant above its initial value $\Phi(\vx_0)$),
    after which point $\Phi$ strictly decreases after each round of this process.
    
    At some point in this process,
    $\calD$ will not be able to reach all the way to the next $\va_i$ without $\Phi$ becoming negative,
    a contradiction.
\end{proof}

The statement of \Cref{thm:non-collapsing_allvoting_entirely_execution_bounded_cannot_decide_majority_parity} does not mention the concept of a leader, 
but it would typically apply to leaderless CRDs.
A CRD may be execution bounded from configurations with a single leader, but not execution bounded when multiple leaders are present (preventing the use of \Cref{thm:potential-function-iff-entirely-execution-bounded}, which requires the CRD to be execution bounded from \emph{all} configurations).
For example, in \Cref{lem:boolean-closure-crd},
reaction \eqref{rxn:boolean-closure-crd-flip-first-vote} occurs finitely many times if the leader/voter $S_Y$ or $S_N$ has count 1.
However, if $S_Y$ and $S_N$ can be present simultaneously (e.g., if we start with two leaders),
then the reactions
$S_Y + V_{NN} \rxn S_Y + V_{YN}$
and
$S_N + V_{YN} \rxn S_N + V_{NN}$
can flip between $V_{NN}$ and $V_{YN}$ infinitely often in an unbounded execution.

If the CRN is leaderless, however, we have the following, which says that if it is execution bounded from \emph{valid initial} configurations, then it is execution bounded from \emph{all} configurations.

\begin{lemmarep}
\label{lem:leaderless-and-execution-bdd-implies-entirely-execution-bdd}
    If a leaderless CRD or CRC is execution bounded,
    then it is entirely execution bounded.
\end{lemmarep}

\begin{proofsketch}
    \opt{submission}{A proof is in the appendix.}
    Since $\calC$ is leaderless, the sum of two valid initial configurations is also valid.
    Thus if we can produce some species from a valid initial configuration,
    we can produce arbitrarily large counts of all species by adding up sufficiently many initial configurations.
    This means that for any configuration $\vx$,
    from any sufficiently large valid initial configuration $\vi$,
    some $\vy \geqq \vx$ is reachable from $\vi$.
    But if $\calC$ is execution bounded from $\vi$,
    since $\vi \reach \vy$,
    it must also be execution bounded from $\vy$,
    thus also from $\vx$ since by additivity any reactions applicable to $\vx$ are also applicable to $\vy$.
\end{proofsketch}

\begin{proof}
    Let $\calC$ be a leaderless CRD or CRC.
    Let $\vx$ be any configuration.
    We first show that some $\vy \geqq \vx$ is reachable from a valid initial configuration $\vi$.
    
    We may assume without loss of generality that $\calC$ only contains species producible from valid initial configurations, otherwise we obtain an equivalent CRN by removing those unproducible species from $\calC$.
    
    Since $\calC$ is leaderless, the sum of two valid initial configurations is also valid.
    Then each species $S$ being producible means that 
    there is a valid initial configuration $\vi_{S,1}$ such that for some $\vy_{S,1}$,
    $\vi_{S,1} \reach \vy_{S,1}$
    and $\vy_{S,1}(S) \geq 1$,
    i.e., at least one copy of $S$ can be produced.
    Let $\vi_{S,k} = k \cdot \vi_{S,1}$.
    By additivity,
    $\vi_{S,k} \reach \vy_{S,k}$,
    where $\vy_{S,k} = k \cdot \vy_{S,1}$,
    noting that $\vy_{S,k}(S) \geq k$.
    In other words, all species are producible in arbitrarily large counts from some valid initial configuration.

    Now we argue all species can be made \emph{simultaneously} arbitrarily large count from some valid initial configuration;
    in particular, we can reach a configuration with counts at least $\vx$.
    Let $\vi = \sum_{S \in \Lambda} \vi_{S,\vx(S)}$.
    Since each $\vi_{S,\vx(S)} \reach \vy_{S,\vx(S)}$,
    by additivity we have
    $\vi \reach \vy$, where $\vy = \sum_{S \in \Lambda} \vy_{S,\vx(S)}$.
    Then for each $S \in \Lambda$, $\vy(S) \geq \vx(S)$,
    so $\vy \geqq \vx$.

    Since all executions from $\vi$ are finite,
    all executions from $\vy$ are finite.
    By additivity,
    any sequence of reactions applicable to $\vx$ is also applicable to $\vy$.
    Thus all executions from $\vx \leqq \vy$ must be finite as well, i.e., $\calC$ is entirely execution bounded since $\vx$ is an arbitrary configuration.
\end{proof}

\Cref{lem:leaderless-and-execution-bdd-implies-entirely-execution-bdd} lets us replace ``entirely execution bounded'' in \Cref{thm:non-collapsing_allvoting_entirely_execution_bounded_decide_almost_constant} with ``leaderless and execution bounded'':

\begin{corollary}
\label{cor:non-collapsing_allvoting_leaderless_execution_bounded_decide_almost_constant}
    If a non-collapsing, all-voting, leaderless, execution bounded CRD stably decides a predicate $\phi$,
    then $\phi$ is eventually constant.
\end{corollary}

In particular, since the original model of population protocols~\cite{angluin2004computation} defined them as leaderless and all-voting---and since population protocols are non-collapsing---we have the following.

\begin{corollary}
\label{cor:pp_execution_bdd_decide_almost_constant}
    If an execution bounded population protocol stably decides a predicate $\phi$,
    then $\phi$ is eventually constant.
\end{corollary}

\begin{toappendix}
    
\subsection{Feedforward CRNs}

We show that another common constraint, \emph{feedforwardness}, significantly reduces computational power, making it impossible to decide even simple mod and threshold sets.

\begin{definition}
\label{defn:reaction-feedforward}
    A CRN is \emph{reaction-feedforward} if reactions can be ordered $r_1, r_2, \ldots, r_n$ such that, for all $k<\ell$, no reactant of $r_k$ appears in $r_{\ell}$ (as either reactant or product).
\end{definition}

Reaction-feedforward CRNs are significant in the sense that many \emph{continuous} real-valued CRNs computing numerical-valued \emph{functions} (where the count of some species $Y$ is interpreted as the output, e.g., $2X \to Y$ computes $f(x) = \lfloor x / 2 \rfloor$) can be computed by reaction-feedforward CRNs~\cite{chen2023rate}.\footnote{
    The definition of feedforward in reference~\cite{chen2023rate} is different from the definition given here, being based on an ordering of species rather than reactions.
    However, it is straightforward to verify by inspection that the CRNs given for the positive results of~\cite{chen2023rate} are reaction-feedforward according to \Cref{defn:reaction-feedforward}.
}
Compared to general CRNs, reaction-feedforward CRNs are easy to analyze and prove correctness. One reason is that, if a reaction-feedforward CRN can reach terminal configuration from $\vx$ at all, then it is execution bounded from $\vx$.

There is a similar definition, called simply \emph{feedforward} in~\cite{chen2023rate}, based on ordering of species rather than reactions.
We use the term \emph{species-feedforward} to avoid confusion with \Cref{defn:reaction-feedforward}.
We say a reaction $(\vr,\vp)$ \emph{produces} a species $S$ if $\vp(S) > \vr(S)$,
and it \emph{consumes} $S$ if $\vr(S) > \vp(S)$.

\begin{definition}
\label{defn:feedforward-species}
    A CRN is \emph{species-feedforward} if species can be ordered $S_1, S_2, \ldots, S_n$ such that every reaction producing a species $S_\ell$ consumes a earlier species $S_k$ where $k<\ell$.
\end{definition}

Although the term ``linear potential function'' was not used in~\cite{chen2023rate},
it is shown in~\cite[Lemma 4.8]{chen2023rate} that species-feedforward CRNs have a linear potential function (assigning weight $\frac{1}{K^i}$ to species $S_i$ for a suitably large constant $K$), thus are entirely execution bounded.
The same is not always true of reaction-feedforward CRNs,
for example $X \rxn 2X$ is reaction-feedforward but not execution bounded.
However, we can use similar techniques to proofs used for so-called \emph{noncompetitive} CRNs in~\cite{vasic2022programming} to show ``reasonable'' reaction-feedforward CRNs are execution bounded.

\begin{lemma}
\label{lem:reaction-feedforward-terminal-only-if-not-less-applications-of-any}
Suppose in a reaction-feedforward CRN that $\vi \reach \vc$ by execution $P$, and $\vi \reach \vd$ by execution $Q$. If any reaction occurs less in $P$ than $Q$, then $\vc$ is not terminal.
\end{lemma}

\begin{proof}
Here we equivalently think of an execution from $\vi$ as a sequence of \emph{reactions}, since from those and $\vi$ we can deduce the configurations in the execution.
Define $\#(r_k, P)$ as the number of times reaction $r_k$ occurs in the execution $P$.
Let $r_k$ be the first reaction in the reaction-feedforward order such that $\#(r_k, P) < \#(r_k, Q)$. Assume, for brevity of explanation, that $r_k$ has only one reactant, denoted $A$;
the argument below, however, is general and applies to any number of reactants in $r_k$.

By the definition of a reaction-feedforward CRN, the reactions $r_{k+1}$ through $r_n$ do not affect the count of $A$. Further, reactions $r_1$ through $r_{k-1}$ can only produce $A$ and not consume it, reactions $r_1$ through $r_k$ can increase the count of $A$, and among them, only $r_k$ can decrease it.
Let $m = \#(r_k, P)$. Let $Q'$ represent the prefix sequence $(\vi, \vx_1, \ldots, \vx_p)$ of $Q$ where the transition $\vx_p \reach \vx_{p+1}$ corresponds to the $(m+1)$st execution of reaction $r_k$. The configuration $\vx_p$ is thus the configuration just before $r_k$ occurs more in $Q$ than in $P$.

Note that reactions $r_1$ through $r_{k-1}$ occur at least as often in $P$ as in $Q$ (i.e. $\#(r_i, P) \geq \#(r_i, Q)$ for $i = 1$ to $k-1$). Therefore, they occur at least as often in $P$ as in $Q'$, since $Q'$ is a prefix of $Q$.
Moreover, by our choice of $Q'$, $\#(r_k, P) = \#(r_k, Q')$. So $A$ is present in $\vc$, i.e. $\vc(A) > 0$. Thus, $r_k$ is applicable at $\vc$, so $\vc$ is not terminal.
\end{proof}

The following corollary implies that any reaction-feedforward CRN that can reach a terminal configuration from $\vi$ is execution bounded from $\vi$.

\begin{corollary}
    In a reaction-feedforward CRN $\calC$, if there is a terminal configuration $\vc_{\vi}$ reachable from initial configuration $\vi$, then $\vc_{\vi}$ is reached by every sufficiently long execution from $\vi$,
    so $\vc_\vi$ is the only terminal configuration reachable from $\vi$.
    Furthermore, all of these executions are permutations of the same number of each reaction type.
    In particular, $\calC$ is execution bounded from $\vi$.
\end{corollary}

\begin{proof}
Let $P$ be the execution leading from $\vi$ to $\vc_{\vi}$. 
Let $Q$ be any execution from $\vi$.

We first claim that $|Q| \leq |P|$;
suppose otherwise the sake of contradiction.
By the pigeonhole principle, $Q$ must have more of some reaction than $P$ does.
By \Cref{lem:reaction-feedforward-terminal-only-if-not-less-applications-of-any}, this would imply that $\vc_{\vi}$ is not terminal, a contradiction.
Therefore, no execution $Q$ can be longer than $P$;
in particular $\calC$ is execution bounded from $\vi$.

If $|Q| = |P|$,
we claim $Q$ must be a permutation of $P$.
Supposing otherwise, by the pigeonhole principle, 
$Q$ would have more of some reaction than $P$, again with \Cref{lem:reaction-feedforward-terminal-only-if-not-less-applications-of-any} implying that $\vc_{\vi}$ is not terminal, a contradiction.

Finally,
suppose $|Q| < |P|$.
By pigeonhole some reaction occurs more frequently in $P$ than in $Q$, so by \Cref{lem:reaction-feedforward-terminal-only-if-not-less-applications-of-any}, $Q$ cannot reach a terminal configuration.
Therefore \emph{all} executions from $\vi$ to $\vc_{\vi}$ are permutations of $P$.
\end{proof}

As noted,
in the model of \emph{continuous} CRNs, it is known that 
all the functions that can be stably computed (the continuous, piecewise linear functions) can be stably computed by reaction-feedforward CRNs~\cite{chen2023rate}.
In contrast, with \emph{discrete} CRNs computing \emph{predicates},
we show that reaction-feedforward CRNs cannot stably decide all semilinear sets by giving two counterexamples, showing that reaction-feedforward CRDs can decide neither ``the simplest'' \emph{mod} set (\cref{lem:feedforward_cannot_decide_mod}) nor ``the simplest'' \emph{threshold} set (\cref{lem:feedforward_cannot_decide_threshold}).
Specifically, we chose the parity and majority predicate as our counterexamples,
although the techniques generalize to more complex mod and threshold sets, e.g., $[X_1 + 2 X_2 \equiv 3 \mod 5?]$.

\begin{lemmarep}
    Reaction-feedforward CRDs can't stably decide the parity predicate $[X \equiv 1 \mod 2?]$.
    \label{lem:feedforward_cannot_decide_mod}
\end{lemmarep}

\begin{proof}
    We show that in any possible construction, the input species must be a reactant of two distinct reactions. By letting the CRN stabilize and then introducing another input molecule, there must exist a set of rules inverting the output in either way, consisting of at least two reactions with $X$ as reactant, breaking the \emph{reaction-feedforward} condition.

    Consider the set of even numbers. A simple, non-reaction-feedforward CRD that decides parity is:
    \begin{align*} 
        Y+X \rightarrow N \\ 
        N+X \rightarrow Y
    \end{align*}
    where $X$ is the input species, $Y$ is a \emph{yes} voter, and $N$ is a \emph{no} voter, initialized with $1Y$ and $nX$. In either way to order these reactions, a reactant of the first reaction appears in the second reaction. Thus, the CRN is not reaction-feedforward.

    To show that no such CRN could decide parity, we show that any construction requires us to have at least one reactant reappear in a later reaction, or even stronger: at least one species must be a reactant of two distinct reactions. Specifically, this is true for the input species $X$.

    To motivate the choice of species, let's consider an even simpler parity computing CRD.
    \begin{align*}
        X+X \rightarrow Y \\
        Y+X \rightarrow X
    \end{align*}
    where $X$ is both input and votes \emph{no}, $Y$ votes yes, initialized with $1Y$ and $nX$.
    Only the input species appears twice as a reactant. Intuitively, this is true for all CRDs because we expect the input to be able to change our answer in either way, reversing the previous one.

    Suppose for the sake of contradiction that there is a reaction-feedforward CRD $\C$ which stably decides whether the initial number $n=\#X$ of input $X$ is even. 
    We withhold two copies of $X$ and let $\C$ stabilize on the correct output of \emph{yes}. 
    Denote $\Upsilon$ as the set of \emph{yes} voters. We denote the \emph{no} voters with $\overline{\Upsilon} \triangleq \Lambda \backslash \Upsilon$. Only species contained in $\Upsilon$ are present in the stable, correct output configuration. 

    Now, we release one of the remaining copies of $X$. 
    We first run the chain reaction (if any) starting from only one $X$. Let $\Omega_{X} := \{ S \mid \exists \vx \in \reachset{\{1X\}}: \vx(S) > 0\}$ be the set of species 
    producible from $\{1X\}$
    (e.g., if there is no reaction $X \rxn ...$, then $\Omega_{X}$ is just $\{X\}$).
    Without loss of generality, we assume $\Omega_{X} \subseteq \Upsilon$, that is, all of $X$'s direct products are \emph{yes-voters} (if not, exchange $\Upsilon$ and $\overline \Upsilon$ in what follows). 
    To correct the answer ($n+1 \not\equiv n \bmod 2$), $\C$ must consume all species currently present and produce at least one copy of a species in $\overline \Upsilon$. It follows that for all $X \in \Omega_{X}$, $\C$ contains a reaction with $X$ as a reactant. Further, none of these reactions contain a reactant of $\overline \Upsilon$, since none are present in the current configuration.
    
    Finally, we release the last remaining copy of $X$. Again, we produce the set $\Omega_{X}$ from $X$. To invert the vote again, we must consume all $Y \in \overline \Upsilon$ and produce at least one member of $\Upsilon$. The reaction(s) consuming $\overline \Upsilon$ must have a member of $\Omega_{X}$ as a reactant since the configuration is stable without $\Omega_{X}$. Further, the reaction cannot be any of the ones from before since they contain a member of $\overline \Upsilon$ as reactant.
    
    Since there are least two reactions sharing a common species as reactant, the reactions cannot be ordered such that no reactant of the first of these reactions appears in the latter one. This makes $\C$ non-reaction-feedforward, contradicting our initial assumption.
\end{proof}

\begin{lemmarep}
    Reaction-feedforward CRDs can't stably decide the majority predicate $[X_1 \geq X_2?]$.
    \label{lem:feedforward_cannot_decide_threshold}
\end{lemmarep}

\begin{proof}
    Suppose, for the sake of contradiction, there exists a reaction-feedforward CRD $\C$ which stably decides the predicate. We let $\C$ stabilize on input $\{n X_1, n X_2\}$ (yielding output \emph{yes}), while withholding two copies of $X_2$. We release one $X_2$. Again, we consider the full set of species a single $X_2$ could produce before reacting with other molecules (denoted $\Omega_{X_2}$). Without loss of generality, we consider all of them \emph{yes} voters i.e. $\Omega_{X_2} \subseteq \Upsilon$. 
    The correct output now changes to \emph{no}, and all \emph{yes} voters must be consumed by reactions that only have reactants which are \emph{yes} voters and further, these reactions contain all species in $\Omega_{X_2}$ as reactants.

    Once the vote has reversed and stabilized, it contains only species of $\overline \Upsilon\triangleq \Lambda \backslash \Upsilon$. We release the last $X_2$ and let it produce $\Omega_{X_2}$. Since $\Omega_{X_2} \subseteq \Upsilon$ i.e. all elements are \emph{yes} voters, but the correct vote is still \emph{no}, all $X \in \Omega_{X_2}$ must be consumed again. This time, they must be consumed in reactions involving \emph{no} voters, must be distinct reactions from those in the previous step. Thus, all species $X \in \Omega_{X_2}$ appear at least twice as a reactant, breaking the reaction-feedforward condition.
\end{proof}










\end{toappendix}

\section{Conclusion}

\opt{full}{
    We explored the computational capabilities of execution bounded Chemical Reaction Networks (CRNs), which terminate after a finite number of reactions. 
    This constraint aligns the model with practical scenarios where fuel supply is limited.
    
    Our findings illustrate that the computational power of these CRNs varies significantly based on structural choices. 
    Specifically, CRNs with an initial leader and the ability to allow only the leader to vote can stably compute all semilinear predicates and functions in $O(\|x\| \log \|x\|)$ parallel time. 
    Without an initial leader, and requiring all species to vote, these networks are limited to computing eventually constant predicates. 
    This limitation holds considerable weight for decentralized systems modeled by population protocols, which inherently exhibit these traits.
    Additionally, we introduced a new characterization of execution bounded networks through a nonnegative linear potential function, providing a novel theoretical tool for analyzing the physical constraints CRNs.
}

A key question remains open: \emph{Can execution bounded CRNs compute semilinear functions and predicates within polylogarithmic time?} Angluin, Aspnes and Eisenstat~\cite{AngluinAE2008Fast} introduced a fast population protocol that simulates a register machine with high probability,
and can be made probability 1 with semilinear predicates. 
However, this construction seems inherently unbounded in executions.

\bibliographystyle{plain}
\bibliography{references}

@article{czerner2024fast,
  title={Fast and succinct population protocols for {P}resburger arithmetic},
  author={Czerner, Philipp and Guttenberg, Roland and Helfrich, Martin and Esparza, Javier},
  journal={Journal of Computer and System Sciences},
  volume={140},
  pages={103481},
  year={2024},
  publisher={Elsevier}
}

@article{vasic2022programming,
  title={Programming and training rate-independent chemical reaction networks},
  author={Vasi{\'c}, Marko and Chalk, Cameron and Luchsinger, Austin and Khurshid, Sarfraz and Soloveichik, David},
  journal={Proceedings of the National Academy of Sciences},
  volume={119},
  number={24},
  pages={e2111552119},
  year={2022},
  publisher={National Acad Sciences}
}

@article{SolCooWinBru08,
  title =	"Computation with Finite Stochastic Chemical Reaction
		 Networks",
  author =	"David Soloveichik and Matthew Cook and Erik Winfree
		 and Jehoshua Bruck",
  journal =	"Natural Computing",
  year = 	"2008",
  number =	"4",
  volume =	"7",
  bibdate =	"2008-11-27",
  bibsource =	"DBLP,
		 http://dblp.uni-trier.de/db/journals/nc/nc7.html#SoloveichikCWB08",
  pages =	"615--633",
  URL =  	"http://dx.doi.org/10.1007/s11047-008-9067-y",
}

@article{Gillespie77,
    author = {Daniel T. Gillespie},
    journal = {Journal of Physical Chemistry},
    number = {25},
    pages = {2340--2361},
    title = {Exact Stochastic Simulation of Coupled Chemical Reactions},
    volume = {81},
    year = {1977}
}

@inproceedings{angluin2004computation,
  title={Computation in networks of passively mobile finite-state sensors},
  author={Angluin, Dana and Aspnes, James and Diamadi, Zo{\"e} and Fischer, Michael J and Peralta, Ren{\'e}},
  booktitle={PODC 2004: Proceedings of the twenty-third annual ACM symposium on Principles of distributed computing},
  pages={290--299},
  year={2004}
}

@article{papadimitriou1981complexity,
  title={On the complexity of integer programming},
  author={Papadimitriou, Christos H},
  journal={Journal of the ACM (JACM)},
  volume={28},
  number={4},
  pages={765--768},
  year={1981},
  publisher={ACM New York, NY, USA}
}

@article{farkas1902theorie,
  title={Theorie der einfachen Ungleichungen.},
  author={Farkas, Julius},
  journal={Journal f{\"u}r die reine und angewandte Mathematik (Crelles Journal)},
  volume={1902},
  number={124},
  pages={1--27},
  year={1902},
  publisher={De Gruyter Berlin, New York}
}

@article{karp1969parallel,
  title={Parallel program schemata},
  author={Karp, Richard M and Miller, Raymond E},
  journal={Journal of Computer and system Sciences},
  volume={3},
  number={2},
  pages={147--195},
  year={1969},
  publisher={Elsevier}
}

@book{mangasarian1994nonlinear,
  title={Nonlinear programming},
  author={Mangasarian, Olvi L},
  year={1994},
  publisher={SIAM}
}

@book{gale1960theory,
  title={The theory of linear economic models},
  author={Gale, David},
  year={1960},
  publisher={University of Chicago press}
}

@article{chen2023rate,
author = {Chen, Ho-Lin and Doty, David and Reeves, Wyatt and Soloveichik, David},
title = {Rate-Independent Computation in Continuous Chemical Reaction Networks},
year = {2023},
issue_date = {June 2023},
publisher = {Association for Computing Machinery},
address = {New York, NY, USA},
volume = {70},
number = {3},
issn = {0004-5411},
url = {https://doi.org/10.1145/3590776},
doi = {10.1145/3590776},
journal = {Journal of the ACM},
month = {May},
articleno = {22},
numpages = {61}
}

@article{Rackoff1978CoveringBoundedness,
  author  = {Charles Rackoff},
  title   = {The covering and boundedness problems for vector addition systems},
  journal = {Theoretical Computer Science},
  volume  = {6},
  number  = {2},
  pages   = {223--231},
  year    = {1978},
  doi     = {10.1016/0304-3975(78)90036-1}
}

@article{Ginsburg1966Semigroups,
  author  = {S. Ginsburg and E. H. Spanier},
  title   = {Semigroups, {P}resburger formulas, and languages},
  journal = {Pacific Journal of Mathematics},
  volume  = {16},
  number  = {2},
  pages   = {285--296},
  year    = {1966}
}

@article{Doty2013Leaderless,
  author = "David Doty and Monir Hajiaghayi",
  title = "Leaderless Deterministic Chemical Reaction Networks",
  journal = "Natural Computing",
  volume={14},
  number={2},
  pages={213-223},
  year=2015,
  note="Preliminary version appeared in DNA 2013."
}

@article{Chen2012DeterministicFunction,
  author =	{Ho-Lin Chen and David Doty and David Soloveichik},
  title =	{Deterministic Function Computation with Chemical Reaction Networks},
  journal = {Natural Computing},
  publisher = {Springer Netherlands},
  year = 2014,
  volume=13,
  number=4,
  pages={517-534},
  doi={10.1007/s11047-013-9393-6},
  note = "Preliminary version appeared in DNA 2012."
}

@article{Ito1969SemilinearSetsFiniteUnionDisjointLinearSets,
title = {Every semilinear set is a finite union of disjoint linear sets},
journal = {Journal of Computer and System Sciences},
volume = {3},
number = {2},
pages = {221-231},
year = {1969},
issn = {0022-0000},
doi = {https://doi.org/10.1016/S0022-0000(69)80014-0},
url = {https://www.sciencedirect.com/science/article/pii/S0022000069800140},
author = {Ryuichi Ito}
}

@article(AngluinAE2008Fast,
title="Fast computation by population protocols with a leader",
author="Dana Angluin and James Aspnes and David Eisenstat",
journal="Distributed Computing",
volume=21,
number=3,
month=sep,
year=2008,
pages={183--199}
)

@article{angluin2007computational,
  title={The computational power of population protocols},
  author={Angluin, Dana and Aspnes, James and Eisenstat, David and Ruppert, Eric},
  journal={Distributed Computing},
  volume={20},
  number={4},
  pages={279--304},
  year={2007},
  publisher={Springer}
}

\end{document}